\documentclass[11pt]{article}

\newcommand{\fullornot}[2]{#1}

\usepackage{geometry} 
\usepackage{amsmath,amsthm,amsfonts,dsfont,mathtools,amssymb}
\usepackage{mathrsfs}
\usepackage{float}
\usepackage{thmtools}
\usepackage{thm-restate}
\usepackage{xcolor}
\usepackage{graphicx}
\usepackage{algorithmic,algorithm}
\usepackage{enumitem}
\setitemize{noitemsep,topsep=3pt,parsep=3pt,partopsep=3pt}

\definecolor{darkgreen}{rgb}{0,0.5,0}
\usepackage{hyperref}
\hypersetup{
    unicode=false,          
    colorlinks=true,        
    linkcolor=red,          
    citecolor=darkgreen,        
    filecolor=magenta,      
    urlcolor=cyan           
}
\usepackage[capitalize, nameinlink]{cleveref}

\crefname{theorem}{Theorem}{Theorems}
\Crefname{lemma}{Lemma}{Lemmas}
\Crefname{figure}{Figure}{Figures}
\Crefname{claim}{Claim}{Claims}
\Crefname{observation}{Observation}{Observations}

\newtheorem{theorem}{Theorem}[section]
\newtheorem{lemma}{Lemma}[section]

\newtheorem{claim}{Claim}[section]
\newtheorem{definition}{Definition}[section]

\newtheorem{observation}{Observation}[section]

\renewcommand{\vec}[1]{\mathbf{#1}}
\newcommand{\nn}{\tilde{n}}
\newcommand{\R}{\mathbb{R}}
\newcommand{\Z}{\mathbb{Z}}
\renewcommand{\mod}{\operatorname{mod}}
\newcommand{\ceil}[1]{\left\lceil #1 \right\rceil}
\newcommand{\floor}[1]{\left\lfloor #1 \right\rfloor}
\newcommand{\paren}[1]{\left( #1 \right)}
\newcommand{\bydef}{:=}
\newcommand{\Expect}{\mathbb{E}}

\newcommand{\ID}{\mathsf{ID}}
\newcommand{\poly}{\operatorname{poly}}
\newcommand{\LOCAL}{\mathsf{LOCAL}}
\newcommand{\CONGEST}{\mathsf{CONGEST}}

\newcommand{\local}{\mathsf{local}}
\newcommand{\Algo}{\mathcal{A}}
\newcommand{\Event}{\mathcal{E}}
\newcommand{\dist}{\mathsf{dist}}
\newcommand{\girth}{\mathsf{girth}}

\newcommand{\eps}[0]{\varepsilon}
\newcommand{\exponential}[0]{\mathsf{Exponential}}
\newcommand{\geom}[0]{\mathsf{Geometric}}
\newcommand{\Bin}[0]{\mathsf{Binomial}}
\newcommand{\Ber}[0]{\mathsf{Bernoulli}}

\usepackage[multiuser,inline,nomargin]{fixme} 
\fxusetheme{color}
\FXRegisterAuthor{z}{ez}{\color{cyan}Z}
\FXRegisterAuthor{y}{ey}{\color{red}Y}
\FXRegisterAuthor{all}{eall}{Note}

\title{The Complexity of Distributed Approximation of \\ Packing  and Covering Integer Linear Programs}

\begin{document}

\date{}
\author{Yi-Jun Chang\footnote{Department of Computer Science, National University of Singapore. Email: cyijun@nus.edu.sg} \and Zeyong Li\footnote{Centre for Quantum Technologies, National University of Singapore. Email: li.zeyong@u.nus.edu}}
\maketitle
\thispagestyle{empty}
\setcounter{page}{0}

\begin{abstract}

\znote{updated abstract}

In this paper, we present a low-diameter decomposition algorithm in the $\LOCAL$ model of distributed computing that succeeds with probability $1 - 1/\poly(n)$. Specifically, we show how to compute an $\left(\epsilon, O\left(\frac{\log n}{\epsilon}\right)\right)$ low-diameter decomposition in $O\left(\frac{\log^3(1/\epsilon)\log n}{\epsilon}\right)$ rounds.

Further developing our techniques, we show new distributed algorithms for approximating general \emph{packing} and \emph{covering} integer linear programs in the $\LOCAL$ model. For packing problems, our algorithm finds an $(1-\eps)$-approximate solution in $O\left(\frac{\log^3 (1/\eps) \log n}{\eps}\right)$ rounds with probability $1 - 1/\poly(n)$. For covering problems, our algorithm finds an $(1+\eps)$-approximate solution in $O\left(\frac{\left(\log \log n + \log (1/\eps)\right)^3 \log n}{\eps}\right)$ rounds with probability $1 - 1/\poly(n)$. These results improve upon the previous $O\left(\frac{\log^3 n}{\eps}\right)$-round algorithm by Ghaffari, Kuhn, and Maus~[STOC 2017] which is based on network decompositions.



Our algorithms are \emph{near-optimal} for many fundamental combinatorial graph optimization problems in the $\LOCAL$ model, such as minimum vertex cover and minimum dominating set, as their $(1\pm \eps)$-approximate solutions require $\Omega\left(\frac{\log n}{\eps}\right)$ rounds to compute. 
\end{abstract}
\section{Introduction}

In this paper, we consider the $\LOCAL$ model~\cite{Linial92} of distributed computing, where a network is modeled as a graph $G$ in such a way that each vertex $v \in V(G)$ corresponds to a computing device and each edge $e \in E(G)$  corresponds to a  communication link. The communication proceeds in synchronous rounds. In each round, each vertex $v \in V(G)$ receives the messages sent from its neighbors, performs some arbitrary local computation, and sends a message of arbitrary size to each of its neighbors. We extend the definition of the  $\LOCAL$ model to hypergraphs $H$ by allowing each vertex $v \in V(H)$ to communicate with all other vertices $u \in V(H)$ such that there is a hyperedge $e\in E(H)$ that contains both $u$ and $v$.

We assume that an upper bound $\nn$ on the number of vertices $n = |V(G)|$ is known to all vertices. All our algorithms work in the setting where a polynomial approximation of the actual number of vertices is known to all devices, i.e., $\nn \leq {|V(G)|}^c$ for some constant $c \geq 1$. All presented lower bounds apply to the setting where the exact number of vertices is known to all devices, i.e., $\nn = n = |V(G)|$.

In the deterministic variant of the model, each vertex $v$ has a distinct identifier of $O(\log n)$ bits. In the randomized variant of the model, each vertex is anonymous and has access to an infinite string of local random bits.

The formulation of general packing and covering integer linear programming (ILP) problems is as follows.

\begin{definition}[Packing problem]\label{def:packing}
    Given $\vec{A}\in \R_{\geq 0}^{m \times n}$, $\vec{b} \in \R_{\geq 0}^{m}$, and $\vec{w} \in \Z_{\geq 0}^{n}$,  find $\vec{x} \in \{0,1\}^{n}$ that maximizes $\vec{w}^\intercal\vec{x}$;
    subject to $\vec{Ax} \leq \vec{b}$.
\end{definition}

\begin{definition}[Covering problem]\label{def:covering}
    Given $\vec{A}\in \R_{\geq 0}^{m\times n}$, $\vec{b} \in \R_{\geq 0}^{m}$, and $\vec{w} \in \Z_{\geq 0}^{n}$, find $\vec{x} \in \{0,1\}^{n}$ that minimizes $\vec{w}^\intercal\vec{x}$;
    subject to $\vec{Ax} \geq \vec{b}$.
\end{definition}

Throughout the paper, we assume that the sum of all weights ${\lVert \vec{w} \rVert}_1 = \sum_{i=1}^{n} w_i$ is polynomial in the number of variables $n$.
Note that although there is a more general formulation of ILP that allows the solution $x_i$ to take values from  non-negative integers in the range $0 \leq x_i \leq s$ and not just $\{0,1\}$, such an ILP instance can be reduced to an instance in our formulation, by decomposing each variable $x_i$ into $\log s$ variables $x_i^{(1)}, \ldots, x_i^{(\log s)}$ taking values in $\{0,1\}$, where $x_i^{(k)}$ represents the $k$th bit of $x_i$.


We consider the following model for integer linear programming in the distributed setting.

\begin{definition}[Modeling of ILP problems]\label{def:hypergraph}
    Given an instance of a packing or covering ILP problem $(\vec{A}\in \R_{\geq 0}^{m \times n}, \vec{b} \in \R_{\geq 0}^{m}, \vec{w} \in \Z_{\geq 0}^{n})$, the hypergraph $H$ associated with the problem is defined by $V(H) := \{x_i\}_{1 \leq i \leq n}$ and $E(H) := \{e_j\}_{1 \leq j \leq m}$, where $e_j := \{x_i : a_{i,j} \neq 0\}$.
\end{definition}

In the above definition, each variable $x_i$ corresponds to a vertex in $V(H)$ and each constraint corresponds to a hyperedge in $E(H)$. Consider the minimum-weight $k$-distance dominating set problem for example, where we are given a network $G$, and the goal is to find a subset of vertices $D \subseteq V(G)$ minimizing $\sum_{v \in D} w(v)$ subject to the condition that $N^k(v) \cap D \neq \emptyset$ for all $v \in V(G)$, where we define $N^k(v) := \{ u \in V(G) : \dist(u,v) \leq k\}$. Its corresponding hypergraph $H$ in \cref{def:hypergraph} is given by $V(H) = V(G)$ and $E(H) = \{N^k(v) : v \in V(G)\}$, and so one round of communication in $H$ can be simulated using $k$ rounds in $G$ in the $\LOCAL$ model.

Modeling ILP problems and other distributed problems as hypergraphs is common. The convenience of the use of hypergraphs for modeling distributed graph problems motivated the study of $\LOCAL$ and $\CONGEST$ algorithms for hypergraphs, see e.g.,~\cite{BEKS19,Balliu23,FischerGK17}. Equivalently, some other works, e.g.,~\cite{GhaffariKM17} used bipartite graphs to model ILPs.

Many fundamental graph problems that are well-studied in theoretical computer science can be formulated as packing and covering ILP, and understanding the complexity of distributed approximation for these problems is a core topic in the area of distributed graph algorithms: maximum matching~\cite{FMU22,LotkerPP08,LPR09}, maximum independent set~\cite{BCGS17,BHKK16}, maximum cut~\cite{BachrachCDELP19,Censor17maxcut}, minimum dominating set~\cite{KUTTEN199840,lenzen2013distributed,amiri2019distributed}, minimum vertex cover~\cite{BarYehudaCS16}, and  minimum edge cover~\cite{goos2013lower}.

\znote{Moving the discussion about LDD earlier:}
\subsection{Low-Diameter Decomposition}
Low-diameter decomposition is an important subroutine for  designing many distributed algorithms, including approximation algorithm of packing and covering ILP problems.

\begin{definition}[Low-diameter decomposition]\label{def:LDD}
    Given a graph $G$, an $(\eps, d)$ low-diameter decomposition is a partition $V(G) = D \cup S_1 \cup S_2 \cup \cdots \cup S_k$ meeting the following conditions.
    \begin{itemize}
        \item  $S_1, S_2, \ldots, S_k$ are mutually non-adjacent subsets.
        \item For each $1 \leq i \leq k$, the weak diameter $\max_{u,v \in S_i}  \dist_G(u,v)$ of $S_i$ is at most $d$.
        \item $D$ contains at most $\eps |V(G)|$ vertices.
    \end{itemize}
\end{definition}

In \cref{def:LDD}, we say that $S_1, S_2, \ldots, S_k$ are the clusters of the decomposition, and each $v \in D$ is called an unclustered vertex or a deleted vertex.
Note that \cref{def:LDD} also applies to hypergraphs.
There is a stronger variant of \cref{def:LDD} that replaces the \emph{weak diameter}  bound with a \emph{strong diameter}  bound: The strong diameter of $S \subseteq V(G)$ is defined as the diameter of the subgraph of $G[S]$ induced by $S$, which is  $\max_{u,v \in S} =\dist_{G[S]}(u,v)$.

It is well-known~\cite{EN16,LinialS93,MPX13} that for any $0 < \eps < 1$, any $n$-vertex graph $G$ admits a low-diameter decomposition with $d = O\left(\frac{\log n}{\eps}\right)$, and such a decomposition can be computed in $O\left(\frac{\log n}{\eps}\right)$ rounds in the $\LOCAL$ model \emph{in expectation} in the sense that the probability that $v$ is unclustered is at most $\eps$, for each $v \in V(G)$.

Bodlaender, Halld\'{o}rsson, Konrad, and Kuhn~\cite{BHKK16} showed that such a low-diameter decomposition algorithm can be used to find an  $(1-\eps)$-approximate maximum independent set \emph{in expectation}, as follows. Compute a  low-diameter decomposition of $G$ with $d = O\left(\frac{\log n}{\eps}\right)$, and then each cluster $S_i$ locally computes a maximum independent set $I_i$. This algorithm can be implemented to run in $O\left(\frac{\log n}{\eps}\right)$ rounds in the $\LOCAL$ model. Clearly $I := \bigcup_{1 \leq i \leq k}I_i$ is an independent set, as  $S_1, S_2, \ldots, S_k$ are mutually non-adjacent. To show that $I$ is an  $(1-\eps)$-approximate maximum independent set in expectation, fix $I^\ast$ to be any maximum independent set of $G$, and by the maximality of $I_i$, we have $|I| = \sum_{i=1}^k |I_i| \geq \sum_{i=1}^k |S_i \cap I^\ast| = |I^\ast| - |D \cap I^\ast|$. Since each $v \in V(G)$ is unclustered with probability  at most $\eps$, we have $\Expect[|D \cap I^\ast|] \leq \eps  |I^\ast|$, so $\Expect[I] \geq (1-\eps)  |I^\ast|$.

\subsection{State of the Art for Packing and Covering ILP}
Kuhn, Moscibroda, and Wattenhofer~\cite{KuhnMW16} showed that an $(1 \pm \eps)$-approximate solution for a general packing and covering \emph{linear program} can be computed in $O\left(\frac{\log n}{\eps}\right)$ deterministically in the $\LOCAL$ model, but their approach does not generalize to \emph{integer linear programs}. The first distributed algorithm that computes an $(1 \pm \eps)$-approximate solution for a general packing and covering ILP in $\poly\left(1/\eps, \log n\right)$ rounds with probability $1 - 1/\poly(n)$ in the $\LOCAL$ model was given by Ghaffari, Kuhn, and Maus~\cite{GhaffariKM17}.

We briefly explain the algorithm of~\cite{GhaffariKM17}, by considering the task of finding an $(1-\eps)$-approximate maximum matching of the network $G$. First, consider the following sequential algorithm, which repeatedly executes the following \emph{ball-growing-and-carving} process for any remaining vertex $v$ until the graph becomes empty.
Let $m_k$ be the size of a maximum matching in the subgraph of $G$ induced by the $k$-radius neighborhood $N^k(v)$ of $v$. If  $k = \Theta\left(\frac{\log n}{\eps}\right)$ is chosen to be sufficiently large, then there exists an index $1 \leq i^\ast < k$ such that $m_{i^\ast} \geq m_{i^\ast + 1} \cdot (1-\eps)$. Fix any maximum matching  in the subgraph of $G$ induced by $N^{i^\ast}(v)$ and remove $N^{i^\ast}(v)$ from the graph. Intuitively, this algorithm finds an $(1-\eps)$-approximate maximum matching  because the cost of not considering the edges between $N^{i^\ast}(v)$ and $V \setminus N^{i^\ast}(v)$ is at most $\eps$ fraction of the size of any fixed maximum matching of $G$ restricted to $N^{i^\ast +1}(v)$.

To give a distributed implementation of the above sequential algorithm in the $\LOCAL$ model, the algorithm of~\cite{GhaffariKM17} uses \emph{network decomposition}. A $(C,D)$ network decomposition of a graph $G$ is a partition of the vertex set $V(G)$ into clusters of diameter at most $D$, where each cluster is assigned a color from $\{1,2, \ldots, C\}$ such that no two adjacent clusters are assigned the same color. It is well-known that an $(O(\log n), O(\log n))$ network decomposition of an $n$-vertex graph exists and can be computed in $O(\log^2 n)$ rounds with probability $1 - 1/\poly(n)$~\cite{LinialS93}.

The algorithm of~\cite{GhaffariKM17} constructs a $(C,D)$  network decomposition of the $k$th power graph $G^{2k}$, where $k = \Theta\left(\frac{\log n}{\eps}\right)$ is the parameter in the above sequential algorithm.  In the original graph $G$, for any two clusters $S_1$ and $S_2$ in the same color class, we must have $\dist(S_1, S_2) \geq k$, so each cluster $S$ in one color class may run the above sequential ball-growing-and-carving process independently in $O(kD)$ rounds, by gathering the entire graph topology of $N^{k}(S)$ to a single vertex in $S$ to simulate the sequential algorithm. Since there are $C = (\log n)$ colors, the overall round complexity is $O(kCD)$.

Using the $O(\log^2 n)$-round randomized $(O(\log n), O(\log n))$ network decomposition algorithm of~\cite{LinialS93}, the algorithm of~\cite{GhaffariKM17} finishes in
$O\left(\frac{\log^3 n}{\eps}\right)$ rounds with probability $1 - 1/\poly(n)$. The algorithm of~\cite{GhaffariKM17} can also be implemented to run in $\poly\left(1/\eps, \log n\right)$ rounds \emph{deterministically} using the recent deterministic polylogarithmic-round network decomposition algorithms~\cite{elkin2022deterministic,GGHIR22,GhaffariGR21,RozhonG20}.

\subsection{Our Contribution}
In the study of low-diameter decompositions, it has been a long-standing open question to make the \emph{in-expectation} guarantee in the algorithms in~\cite{EN16,LinialS93,MPX13} to hold \emph{with high probability}:
\begin{description}
    \item[(C1)] Design an algorithm that finds an  $\left(\epsilon, O\left(\frac{\log n}{\epsilon}\right)\right)$ low-diameter decomposition in $\tilde{O}\left(\frac{\log n}{\epsilon}\right)$ rounds such that the bound $|D| \leq \eps|V(G)|$ on the number of unclustered vertices holds with probability $1 - 1/\poly(n)$, where the notation $\tilde{O}(\cdot)$ hides polylogarithmic factors in the sense that $\tilde{O}(f(n,\eps))= f(n,\eps) \cdot \log^{O(1)}f(n,\eps)$.
\end{description}

In this paper, we present a new low-diameter decomposition algorithm that resolves (C1), at the cost of increasing the round complexity by a small $\poly\left(\log (1/\eps)\right)$ factor. We prove the following theorem in \cref{sect:LDD-main}.

\begin{theorem}\label{thm:LDD}
There exists an algorithm that computes an $\left(\eps, O\left(\frac{\log n}{\eps}\right)\right)$ low-diameter decomposition in the $\LOCAL$ model in $O\left(\frac{\log^3(1/\eps) \log n}{\eps}\right)$ rounds with probability $1 - 1/\poly(n)$.
\end{theorem}

Note that (C1) is relevant to many distributed algorithms that are built on low-diameter decompositions. For example, Elkin and Neiman showed that a spanner of stretch $2k - 1$ and \emph{expected} size $O\left(n^{1+1/k}\right)$ can be computed in $O(k)$ rounds in the $\CONGEST$ model~\cite{elkin2018efficient}, which is a variant of the $\LOCAL$ that restricts the number of bits that can be transmitted along each edge to be $O(\log n)$. The bound on the size of spanner holds only in expectation because it is built on the low-diameter decomposition algorithm of~\cite{MPX13}, and it remains an open question whether such a bound can be achieved with probability $1 - 1/\poly(n)$, see~\cite{forsterOPODIS2021}. Another example is the expander decomposition algorithm of Chang and Saranurak~\cite{changS19}, which also uses low-diameter decompositions. To ensure that the guarantee on the number of inter-cluster edges in the expander decomposition holds with probability $1 - 1 / \poly(n)$ and not only in expectation, a new low-diameter decomposition algorithm was designed in~\cite{changS19}.
Specifically, it was shown in~\cite{changS19} that there is a $\poly\left(1/\eps, \log n\right)$-round algorithm that finds an $\left( \eps, O\left(\frac{\log^2 n}{\eps^2}\right)\right)$ low-diameter decomposition with probability $1 - 1 / \poly(n)$ in the $\CONGEST$ model.
In \fullornot{\cref{sect:LDD}}{the full version \cite[Appendix C.1]{fullversion}}, we show that there exists a family of graphs such that if we run  existing low-diameter decomposition algorithms on these graphs, then with non-negligible probability, the number of unclustered vertices exceed $\eps |V(G)|$, so (C1) is not merely an issue of analysis and we really need a new low-diameter decomposition algorithm.

\paragraph{Packing and Covering ILP.} Building on the techniques behind our new low-diameter decomposition algorithm, we construct new distributed algorithms for solving general packing and covering integer linear programs in the $\LOCAL$ model of distributed computing. We show that $(1\pm \eps)$-approximate solutions of  these problems can be computed in $\tilde{O}\left(\frac{\log n}{\eps}\right)$ rounds with probability $1 - 1/\poly(n)$.
Our result improves  upon the previous $O\left(\frac{\log^3 n}{\eps}\right)$-round algorithm of~\cite{GhaffariKM17}, bypassing the $O\left(\log^2 n\right)$ barrier for algorithms based on network decompositions, and getting closer to the bound $O\left(\frac{\log n}{\eps}\right)$ for the case fractional solutions are allowed~\cite{KuhnMW16}. \fullornot{We prove the following theorems in \cref{sect:packing-main,sect:covering-main}}{See the full version \cite[Sections 4 and 5]{fullversion} for the proof of the following theorems}.

\begin{theorem}[Algorithm for packing problems]\label{thm:packing}
There is an algorithm that computes a $(1-\eps)$-approximate solution for any packing integer linear programming problem in the $\LOCAL$ model in $O\left(\frac{\log^3(1/\eps) \log n}{\eps}\right)$ rounds with probability $1 - 1/\poly(n)$.
\end{theorem}

\begin{theorem}[Algorithm for covering problems]\label{thm:covering}
There is an algorithm that computes a $(1+\eps)$-approximate solution for any covering integer linear programming problem in the $\LOCAL$ model in $O\left(\frac{\left(\log\log n + \log(1/\eps)\right)^3 \cdot \log n}{\eps}\right)$ rounds with  probability  $1 - 1/\poly(n)$.
\end{theorem}

 Our results  imply that for many well-studied fundamental distributed problems, such as maximum independent set, maximum cut, minimum vertex cover, minimum dominating set, and many of their variants, $(1 \pm \eps)$-approximate solutions  can be found in  $\tilde{O}\left(\frac{\log n}{\eps}\right)$ rounds with probability $1 - 1/\poly(n)$ in the $\LOCAL$ model. To the best of our knowledge, such algorithms for these problems were not known prior to our work.
 This upper bound is \emph{nearly tight} in that   for several problems, there is an $\Omega\left(\frac{\log n}{\eps}\right)$ lower bound for computing an $(1\pm \eps)$-approximate solution.
 

\begin{theorem}[Lower bounds]\label{thm:lowerbound-main}
The following lower bound holds for these problems:
\begin{itemize}
    \item maximum independent set,
    \item maximum cut,
    \item minimum vertex cover,
    \item minimum dominating set.
\end{itemize}
There is a universal constant $0 < \eps_0 < 1$ such that for any randomized algorithm $\Algo$ in the $\LOCAL$ model whose expected value of  solution is within an $1\pm \eps$ factor to the optimal solution, with $0 < \eps \leq \eps_0$,
 the round complexity of $\Algo$ is $\Omega\left(\frac{\log n}{\eps}\right)$.
\end{theorem}

The lower bounds of \cref{thm:lowerbound-main} apply to randomized algorithms whose approximation guarantee only holds \emph{in expectation}. By standard reductions, these lower bounds also apply to randomized algorithms that succeed \emph{with high probability} and deterministic algorithms.

We emphasize that the $\Omega\left(\frac{\log n}{\eps}\right)$ lower bounds of \cref{thm:lowerbound-main} and their proofs are not entirely new.
For the \emph{minimum vertex cover} problem, an  $\Omega(\log n)$ lower bound for constant approximation was shown by G{\"{o}}{\"{o}}s and Suomela~\cite{GoosS14}. By subdividing edges into degree-2 paths, this lower bound was extended to $\Omega\left(\frac{\log n}{\eps}\right)$ for $(1+\eps)$ approximation by Faour, Fuchs and Kuhn~\cite{FFK21}.  For the \emph{maximum independent set} problem, an  $\Omega(\log n)$ lower bound for constant approximation was shown
 by Bodlaender, Halld\'{o}rsson,   Konrad, and Kuhn~\cite{BHKK16}, and a similar   $\Omega\left(\frac{\log n}{\eps}\right)$ lower bound  for $(1-\eps)$ approximation was shown by  Balliu, Kuhn, and Olivetti~\cite{BKO21} in the context of fractional coloring.
We still include a proof of \cref{thm:lowerbound-main} in \fullornot{\cref{sect:lowerbounds}}{the full version \cite[Appendix B]{fullversion}} for the sake of completeness.

\subsection{Our Method}\label{sect:method}
In the subsequent discussion, we say that an event happens \emph{with high probability} if it happens with probability at least $1 - n^{-c}$ for some suitably large constant $c \geq 1$.

As discussed earlier, a low-diameter decomposition algorithm that works \emph{in expectation} can be used to approximately solve packing and covering ILPs \emph{in expectation}. To turn such an algorithm into one that successfully computes an
$(1\pm\eps)$-approximate solution with high probability, we need the following as the first step:
\begin{description}
    \item[(C1)] Design an algorithm that finds an  $\left(\epsilon, O\left(\frac{\log n}{\epsilon}\right)\right)$ low-diameter decomposition in $\tilde{O}\left(\frac{\log n}{\epsilon}\right)$ rounds such that the bound $|D| \leq \eps|V(G)|$ on the number of unclustered vertices holds with probability $1 - 1/\poly(n)$. 
\end{description}

In addition, we need to overcome the following challenge:
\begin{description}
    \item[(C2)] For any fixed optimal solution $I^\ast$ to the considered ILP, which is \emph{unknown} to the algorithm, we need to strengthen the guarantee $|D| \leq \eps|V(G)|$ to $|D \cap I^\ast| \leq \eps |I^\ast|$, and this also needs to hold with high probability.
\end{description}

The main idea behind our solution to (C1) is a new \emph{graph sparsification} algorithm based on iterative ball-growing-and-carving with random choices of centres. We will show that after applying the sparsification procedure, the remaining part of the graph is sufficiently sparse that if we run the existing low-diameter decomposition algorithms~\cite{EN16,LinialS93,MPX13}, then the guarantee on the number of unclustered vertices holds with probability $1 - 1/\poly(n)$, due to a Chernoff bound for variables with limited dependence~\cite{Pem01}.
For (C2), we develop a new \emph{sampling technique} based on pre-computing $O(\log n)$ independent low-diameter decompositions that allows us to simulate sampling from an \emph{arbitrary fixed unknown optimal solution} for the underlying ILP problem. In \cref{sect:technicaloverview-1,sect:technicaloverview-2,sect:technicaloverview-3}, we present a  technical overview of our solutions and how they lead to the proofs of \cref{thm:packing,thm:covering,thm:LDD}.

\subsubsection{Low-Diameter Decompositions}\label{sect:technicaloverview-1}

\znote{some changes made:}

Our starting point is a Chernoff bound for variables with limited dependence \cite{Pem01}. If we were able to show that the event of a vertex being unclustered is dependent on at most $O\left(\frac{\eps|V|}{\log n}\right)$ other such events, we can derive that with probability $1 - 1/\poly(n)$, the number of unclustered vertices is within a constant factor of its expectation $\mu \leq \eps|V(G)|$.

For any $k$-round algorithm in the $\LOCAL$ model, if two vertices $u$ and $v$ satisfy $\dist(u,v) \geq 2k+1$, then the local output of $u$ and the local output of $v$ are independent. As the low-diameter decomposition algorithms of~\cite{EN16,LinialS93,MPX13} take $k = O\left(\frac{\log n}{\eps}\right)$ rounds, it suffices to ensure that the $2k$-radius neighborhood $N^{2k}(v)$ of each vertex $v \in V(G)$ contains at most $O\left(\frac{\eps |V(G)|}{\log n}\right)$ vertices.

For ease of presentation, here we will assume that $n = |V(G)|$ is known to all vertices. Later we will discuss how to extend our algorithm to the case that only a polynomial upper bound $\nn \geq n$ is known.  Our sparsification algorithm consists of $O(\log (1/\eps))$ iterations of ball-growing-and-carving. Specifically, we set $t = \log(1/\eps) + O\left(1\right)$ and $R = O\left({\frac{t\log n }{\eps}}\right)$, and we decompose the interval $[R+1, (t+2)R]$ into $t+1$ intervals $I_{t+1}, I_{t},\ldots, I_1$ of size $R$ by setting $I_i := [a_i, b_i] =[(t-i+2)R+1, (t-i+3)R]$. We select the parameter $R$ in such a way that $R \geq 4k$, where $k = O\left(\frac{\log n}{\eps}\right)$ is the round complexity of an existing low-diameter decomposition, as we may freely assume $\eps$ is at most some universal constant $0 < \eps_0 < 1$.

\paragraph{Phase 1.} The first phase of our algorithm consists of $t$ iterations of ball-growing-and-carving. Specifically, for $1 \leq i \leq t$, in the $i$th iteration, each vertex $v$ declares that it is a centre with probability $p_i = \Theta\left(\frac{2^i \log n}{n}\right)$. For each centre $v$, it runs the following ball-growing-and-carving procedure: Find $j^\ast \in I_i$ that minimizes $|N^{j^\ast}(v) \setminus N^{j^\ast-1}(v)|$, \emph{delete}  the vertices in $N^{j^\ast}(v) \setminus N^{j^\ast-1}(v)$, and \emph{remove} the vertices in $N^{j^\ast}(v)$ that are not deleted by the ball-growing-and-carving procedure due to other centres.
The removed vertices form connected components of weak diameter at most $2b_i = O(t\cdot R) = O\left(\frac{\log^2(1/\eps) \cdot \log n}{\eps}\right)$, so they are considered clustered.
The deleted vertices are considered unclustered and they will not be considered in the subsequent steps of the algorithm.

We first show that the algorithm indeed sparsifies the graph.
Due to our choice of $p_i$, by a Chernoff bound and a union bound, we may show that, with high probability, the $b_{i+1}$-radius neighborhood of each vertex $v \in V(G)$ that still remain in the graph after the $i$th iteration contains at most  $O\left(\frac{n}{2^i}\right)$ vertices. The reason is that if the $b_{i+1}$-radius neighborhood of $v$ contain  $\omega\left(\frac{n}{2^i}\right)$ vertices after the $i$th iteration, then the $a_{i}$-radius neighborhood of $v$ at the beginning of the $i$th iteration also contains at least $\omega\left(\frac{n}{2^i}\right)$ vertices, as $a_i > b_{i+1}$, and so the expected number of centres in the $a_{i}$-radius neighborhood of $v$ during the $i$th iteration is at least $\omega\left(\frac{n}{2^i}\right) \cdot p_i = \omega(\log n)$, meaning that $v$ is either removed or deleted during the $i$th iteration with high probability.

We also need to show that the number of deleted vertices is small.
We know that at the beginning of the $i$th iteration, with high probability, the $b_{i}$-radius neighborhood of each vertex $v \in V(G)$  contains at most  $O\left(\frac{n}{2^{i-1}}\right)$ vertices, so the expected number of centres $u$ such that $v \in N^{b_i}(u)$ during the $i$th iteration is $O(\log n)$. This bound also holds with high probability, due to a Chernoff bound. For each centre $u$, the number of vertices deleted due to the ball-growing-and-carving procedure of $u$ is at most $1/R$ fraction of the size of $N^{b_i}(u)$, so we conclude that the total number of deleted vertices in the $i$th iteration is at most $n \cdot O(\log n) \cdot \frac{1}{R} = O\left(\frac{\eps n}{t}\right)$. By choosing the leading constants in the $O(\cdot)$-notation properly, we can make sure that the number of vertices deleted during the $t$ iterations is at most a small constant fraction of $\eps n$.

\paragraph{Phase 2.} After finishing the algorithm of Phase 1,  the $b_{t+1}$-radius neighborhood of each vertex $v \in V(G)$  contains at most  $O\left(\frac{n}{2^{t}}\right) = O\left(\eps n\right)$ vertices. Recall that our goal is to sparsify the graph so that the size of the $R$-radius neighborhood for each remaining vertex is $O\left(\frac{\eps n}{\log n}\right)$. One strategy to achieve this goal is to simply set the number of iterations in Phase 1 to be $t = O(\log(1/\eps) + \log \log n)$, but there is a more efficient solution: Our algorithm of Phase 2 consists of just  \emph{one  iteration} of ball-growing-and-carving using the sampling probability $p_{t+1} = \Theta\left(\frac{2^{t+1} \log n \log (1/\eps)}{n}\right)$ and the interval $I_{t+1} = [R+1, 2R]$.

To show that one iteration with the above choice of sampling probability suffices, in the analysis, we consider an arbitrary partition of the graph at the end of Phase 1 into several \emph{dense} clusters and one \emph{sparse} part with the following properties.
\begin{itemize}
    \item Each dense cluster has weak diameter at most $R$ and contains $\Theta\left(\frac{\eps n}{\log n}\right)$ vertices. 
    \item For each vertex in the sparse part, its $(R/2)$-radius neighborhood contains  $O\left(\frac{\eps n}{\log n}\right)$ vertices (note that we only count neighbours in the sparse part).
\end{itemize}

Intuitively, if a vertex $v$ belongs to a sparse part, then its local sparsity is already good enough, so we just need to bound the number of vertices in the dense part that are not removed or deleted in Phase 2. These vertices are called \emph{bad vertices}.
By our choice of $p_{t+1}$ and $I_{t+1}$, for each dense cluster, with probability at least $1 - \eps$ the vertices in the entire cluster are removed or deleted during  Phase 2.
By a Chernoff bound, we may show that, with high probability, the  number of dense clusters whose members are not completely removed or deleted is $O(\log n)$, so the number of bad vertices  is at most a small constant fraction of $\eps n$.

\paragraph{Phase 3.} In the last step of our algorithm, we apply the existing low-diameter decomposition of~\cite{EN16} to the subgraph induced by the remaining vertices with a parameter $\eps'$ that is a small constant fraction of $\eps$.  After Phase 2, if a remaining vertex $v$ is not bad, then its $(R/2)$-radius neighborhood has at most $O\left(\frac{\eps n}{\log n}\right)$ vertices. As $R/2 \geq 2k$, by a Chernoff bound for variables with limited dependence, we know that if we apply an existing low-diameter decomposition to the current graph, then the guarantee on the number of unclustered vertices hold with high probability.

\paragraph{Summary.} The above algorithm computes a low-diameter decomposition such that each cluster has weak diameter $O\left(\frac{\log^2(1/\eps) \cdot \log n}{\eps}\right)$ and the bound $\eps|V(G)|$ on the number of unclustered vertices holds \emph{with high probability}. We may apply a brute-force computation for each cluster to improve the diameter bound to the ideal one, which is $O\left(\frac{\log n}{\eps}\right)$. The round complexity of our algorithm is  $O(t^2\cdot R) = O\left(\frac{\log^3(1/\eps) \cdot \log n}{\eps}\right)$. To extend this algorithm to a more realistic setting where each vertex only knows a polynomial approximation $\nn$ of the actual number $n = |V(G)|$ of vertices, we may simply let each vertex $v$ compute its estimate $n_v$ of $n$ by $n_v := |N^{4tR}(v)|$ and change its sampling probability to $p_{v,i} = \Theta\left(\frac{2^i \log \nn}{n_v}\right)$. Intuitively, this works because $4tR$ is large enough to cover all vertices that are relevant to $v$ throughout the algorithm.

\paragraph{Remark.} Note that the low-diameter decomposition algorithm of~\cite{changS19} mentioned earlier is also based on a Chernoff bound for variables with limited dependence. The strategy of~\cite{changS19} is to directly compute a  partition $V = V_D \cup V_S$ in such a way that each connected component of $V_D$ already has small diameter and $V_S$ is sufficiently sparse that we may apply a Chernoff bound for variables with limited dependence. The algorithm of~\cite{changS19} inherently needs significantly more than $O\left(\frac{\log n}{\eps}\right)$ rounds, as the diameter of a connected component of $V_D$ can be as large as $O\left(\frac{\log^2 n}{\eps^2}\right)$, so their method is unsuitable for us, as we aim for a round complexity that is nearly $O\left(\frac{\log n}{\eps}\right)$.

\znote{made a few modifications in the packing and covering parts:}

\subsubsection{Packing Problems}\label{sect:technicaloverview-2}

For  ease of presentation, here we take the maximum independent set problem as an example, and we let $I^\ast \subseteq V(G)$ denote an arbitrary fixed maximum independent set of the input graph $G$. Our goal is to modify the approach presented in \cref{sect:technicaloverview-1} so that, with high probability, at most $\eps$ fraction of $I^\ast$ are unclustered in the low-diameter decomposition computed by the distributed algorithm. As discussed earlier, given such a low-diameter decomposition, an $(1-\eps)$-approximate maximum independent set can be obtained by taking the union of the maximum independent $I_j$ of cluster $S_j$, over all clusters $S_1, S_2, \ldots, S_k$ in the low-diameter decomposition.

\znote{added:}

Following the paradigm of our low-diameter decomposition algorithm, it suffices to ``sparsify'' the graph such that any $O(R)$-radius neighbourhood contains not too many vertices in $I^\ast$, while the    ``sparsification'' process should not remove too many vertices in $I^\ast$. That is, everything should be measured against the number of vertices in $I^\ast$, instead of the number of the vertices as in the low-diameter decomposition case. 

If $I^\ast$ is \emph{known} to the algorithm, then we may simply  run the algorithm of \cref{sect:technicaloverview-1} restricted to $I^\ast$. That is, only the vertices in $I^\ast$ may sample themselves to be the centres, and when they do ball-growing-and-carving, they only count the vertices in $I^\ast$ when they decide which layer to cut.

To deal with the issue that $I^\ast$ is \emph{unknown} to the algorithm, we need to have a rough estimate of $I^\ast$. More specifically, we need to roughly know which neighbourhood of the graph contains a large number of vertices from $I^\ast$ so that we can handle them via the ball-growing-and-carving procedure. Towards achieving this, we will do a \emph{preparation step} that computes $\Theta(\log \nn)$ independent low-diameter decompositions, using \cite{EN16} and not our algorithm, with $\eps' = \frac{1}{2}$. we write $\mathcal{C}$ to denote the set of all clusters in these  $\Theta(\log \nn)$ low-diameter decompositions. We emphasize that each of these clusters would operate fully independently throughout the algorithm. Similar to the computation of the estimate $n_v$, each cluster $C \in \mathcal{C}$  will compute an estimate of the independence number by calculating the size $W(P^{\local}_{S_C},S_C)$ of a maximum independent set $P^{\local}_{S_C}$ of the subgraph induced by its neighbourhood $S_C := N^{8tR}(C)$. Here the $P^{\local}_{U}$ denotes an optimal solution of the underlying packing problem $P$ restricted to the subset $U$, and $W(P,U)$ denotes the weight of the solution $P$ restricted to the variables in $U$.
Each cluster $C \in \mathcal{C}$  will also calculate its own weight by the size $W(P^{\local}_C,C)$ of a maximum independent set $P^{\local}_C$ of the subgraph induced by $C$. Note that this is an overestimate of $|I^\ast \cap C|$.

With these estimates in hand, we can now modify and run the algorithm of \cref{sect:technicaloverview-1} with respect to $I^\ast$. In particular, we will do ball-growing-and-carving from these clusters and not from individual vertices, and the sampling probability for each cluster $C$ in iteration $i$ is set to $p_{C,i} = \Theta\left(\frac{2^i W(P^{\local}_C,C)}{W(P^{\local}_{S_C},S_C)}\right)$, which is the weight of $C$ divided by the estimate computed by $C$ and multiplied by $2^i$. Intuitively, this measures the relative ``importance'' of this cluster in computing the independence set. We need $\Theta(\log \nn)$ independent low-diameter decompositions to ensure that this measure is close to the true value on average. As a result of using $\Theta(\log \nn)$ low-diameter decompositions, the $\Theta(\log \nn)$-factor in the definition of $p_{v,i}$ in \cref{sect:technicaloverview-1} is also removed.

We will show that this new sampling method approximates the restriction of the algorithm of \cref{sect:technicaloverview-1} to $I^\ast$ well, so we may show that, with high probability, at most $\eps$ fraction of $I^\ast$ are unclustered in the low-diameter decomposition computed by the algorithm with the new sampling method.

\subsubsection{Covering Problems}\label{sect:technicaloverview-3}

For covering problems, the $(1-\eps)$-approximate maximum independent set algorithm based on low-diameter decomposition discussed earlier does not work, as we may not simply set the variables corresponding to the unclustered vertices to zero. For example, consider the minimum dominating set problem. Suppose there is a vertex $v$ adjacent to $s$ vertices $u_1, u_2, \ldots, u_s$ of degree one. If $s$ is unclustered and each of $u_1, u_2, \ldots, u_s$ is a singleton cluster, then the cost of not letting $v$ join the dominating set is that all of $u_1, u_2, \ldots, u_s$ have to join the dominating set.

To deal with this issue, we will consider a variant of low-diameter decomposition of a hypergraph $H$, which can be obtained by modifying the  algorithm  and analysis of~\cite{MPX13}. In this variant of low-diameter  decomposition, the goal is to find a set of clusters $S_1, S_2, \ldots, S_k$ with small weak diameter satisfying the following conditions.
\begin{itemize}
    \item For each hyperedge $e \in E(H)$, there exists $1 \leq i \leq k$ such that $e$ completely covered in $S_i$.
    \item For each vertex $v \in V(H)$, the number of clusters that contains $v$ is dominated by a geometric random variable with parameter $e^{-\eps}$.
\end{itemize}
That is, $\{S_1, S_2, \ldots, S_k\}$ is a \emph{sparse cover} of all hyperedges, and it is sparse in the sense that for each vertex $v \in V(H)$, the expected number of clusters that contains $v$ is at most $\frac{1}{e^{-\eps}} \approx 1+\eps$.

There will be two main differences between our algorithms for packing problems and covering problems.
The first difference is that for packing, we will replace the low-diameter decompositions in the preparation step and the last step of the algorithm in \cref{sect:technicaloverview-2} with the variant discussed above. The second difference is that due to the reason that we cannot handle unclustered vertices, we cannot  tolerate the bad vertices due to Phase 2 in \cref{sect:technicaloverview-1,sect:technicaloverview-2}, so we will have to skip Phase 2 by increasing the number of iterations of Phase 1 from $t = O(\log(1/\eps))$ to $t = O(\log(1/\eps) + \log \log n)$, causing the round complexity of \cref{thm:covering} to be higher than that of \cref{thm:packing,thm:LDD}.

\subsection{Additional Related Work}
Awerbuch, Goldberg, Luby, and Plotkin~\cite{awerbuch89} presented  the first distributed algorithm for network decompositions: They showed that a network decomposition with $2^{O(\sqrt{\log n\cdot \log\log n})}$ colors and strong diameter $2^{O(\sqrt{\log n\cdot \log\log n})}$ can be computed in  $2^{O(\sqrt{\log n\cdot \log\log n})}$ rounds \emph{deterministically} in the $\CONGEST$ model. 

The bounds $2^{O(\sqrt{\log n\cdot \log\log n})}$ were later improved to $2^{O(\sqrt{\log n})}$ by Panconesi and Srinivasan~\cite{panconesi-srinivasan}, but their   algorithm works only in the $\LOCAL$ model as it requires messages of unbounded size.
The message size bound was improved to  $O(\log n)$ by Ghaffari~\cite{ghaffari2019MIS}, and then the algorithm was extended to power graphs by Ghaffari and Portmann~\cite{ghaffari_et_al:LIPIcs:2019:11325}.

In the randomized setting,
 Linial and Saks~\cite{LinialS93}
 gave the first $O\left(\frac{\log n}{\eps}\right)$-round  algorithm for  low-diameter decomposition with \emph{weak diameter} $O\left(\frac{\log n}{\eps}\right)$, where the bound $\eps|V(G)|$ on the number of unclustered vertices only hold \emph{in expectation}. This implies an $O\left(\log^2 n\right)$-round algorithm for network decomposition with $O(\log n)$ colors and $O(\log n)$ {weak diameter} in the $\CONGEST$  model. Based on the techniques developed by Miller, Peng, and Xu~\cite{MPX13}, Elkin and Neiman~\cite{EN16} improved these results to \emph{strong-diameter} decompositions.

 In a breakthrough result, Rozho\v{n} and Ghaffari~\cite{RozhonG20} showed that a low-diameter decomposition with \emph{weak diameter} $O\left(\frac{\log^3 n}{\eps}\right)$ can be constructed \emph{deterministically} in $O\left(\frac{\log^6 n}{\eps^2}\right)$ rounds the $\CONGEST$ model, giving the first \emph{polylogarithmic-round deterministic} network decomposition algorithm. The algorithm was extended to power graphs by Maus and Uitto~\cite{DBLP:conf/wdag/MausU21}.

 The Rozho\v{n}--Ghaffari algorithm was subsequently improved by a series of works~\cite{chang2021strong,elkin2022deterministic,GGHIR22,GhaffariGR21}, leading to an $\tilde{O}\left(\frac{\log^2 n}{\eps} \right)$-round deterministic low-diameter decomposition algorithm with \emph{strong diameter} $\tilde{O}\left(\frac{\log n}{\eps}\right)$ in the $\CONGEST$ model~\cite{GGHIR22}, which is obtained by partially derandomizing the randomized algorithm of~\cite{MPX13}.



\paragraph{Applications.}
Low-diameter decomposition is a very useful building block in designing  distributed algorithms, and it has found applications to distributed approximation~\cite{BHKK16,amiri2019distributed,Censor17maxcut,czygrinow2020distributed,CZYGRINOW2006JDA, Czygrinow2006ESA,Czygrinow2007cocoon,czygrinow2008fast,FFK21,faour2021approximating}, distributed property testing~\cite{Even17,levi2021property},  distributed spanner constructions~\cite{elkin2018efficient,forsterOPODIS2021}, distributed densest subgraph detection~\cite{su_et_al:LIPIcs:2020:13093}, and radio network algorithms~\cite{Chang20bfs,Chang18broadcast,CzumajD17,dani2022wake,haeupler2016faster}.

In graphs \emph{excluding a fixed minor}, a low-diameter decomposition can be computed in $\poly(1/\eps) \cdot O(\log^\ast n)$ rounds \emph{deterministically} in the $\LOCAL$ model~\cite{czygrinow2008fast}. There was a long line of works utilizing low-diameter decompositions to design efficient approximation algorithms in graphs excluding a fixed minor in the $\LOCAL$ model~\cite{amiri2019distributed,czygrinow2020distributed,CZYGRINOW2006JDA, Czygrinow2006ESA,Czygrinow2007cocoon,czygrinow2008fast}. This line of research was recently extended to the $\CONGEST$ model~\cite{chang2022narrowing,chang2023efficient}.
The \emph{randomized} low-diameter decompositions of~\cite{LinialS93,EN16,MPX13} were utilized to give approximation algorithms for maximum independent set~\cite{BHKK16} and maximum cut~\cite{Censor17maxcut}. 
The two recent works~\cite{FFK21,faour2021approximating} also utilized the low-diameter decomposition of~\cite{MPX13} to design $(1\pm \eps)$-approximation algorithms for maximum matching, minimum vertex cover, and their weighted versions. Due to the use of~\cite{LinialS93,EN16,MPX13}, same as ~\cite{BHKK16,Censor17maxcut}, the approximation ratio guarantee of the randomized algorithms in~\cite{FFK21,faour2021approximating} holds \emph{in expectation}. Note that the main focus of~\cite{FFK21,faour2021approximating} is to obtain efficient algorithms in the $\CONGEST$ model. In comparison, the focus of our work is to design efficient algorithms that are applicable to any packing and covering ILP with high probability guarantees in the $\LOCAL$ model.  

Low-diameter decomposition can also be used to construct \emph{network decompositions} and \emph{expander decompositions}, which have numerous applications in designing distributed graph algorithms. In particular, it is known~\cite{RozhonG20,GhaffariKMU18} that any \emph{sequential local} polylogarithmic-round \emph{randomized} algorithm can be converted into a \emph{deterministic} polylogarithmic-round distributed algorithm in the $\LOCAL$ model, via network decompositions.
Distributed expander decomposition has many applications to distributed subgraph finding problems~\cite{changS19,ChangS20,CensorLL20,censor2021tight,EFFKO19,izumiLM20,LeGall2021} in  $\CONGEST$. 

\subsection{Subsequent Work}
After the publication of the conference version~\cite{chang2023complexity} of this work,
Coiteux-Roy~et~al.~\cite{bootstrapping}
 presented a blackbox construction of a $\left(\eps, O\left(g(n)/\eps\right) \right)$ low-diameter decomposition algorithm that runs in $O\left(\left(f(n)+g(n)\right)\cdot \frac{\log (1/\eps)}{\eps}\right)$ rounds, when given a $\left(\frac{1}{2}, g(n)\right)$ low-diameter decomposition algorithm that runs in $f(n)$ rounds~\cite[Theorem 3.10]{bootstrapping}. By choosing $\eps = \frac{1}{2}$ in \cref{thm:LDD}, we may use $g(n) = f(n) = O(\log n)$, so the blackbox construction yields a $\left(\eps, O\left(\frac{\log n}{\eps}\right) \right)$ low-diameter decomposition algorithm that runs in $O\left(\frac{\log (1/\eps)\log n}{\eps}\right)$ rounds with probability $1 - 1/\poly(n)$ in the $\LOCAL$ model. That is, the $\log^3 (1/\eps)$ factor in the round complexity of \cref{thm:LDD} can be improved to $\log (1/\eps)$.

For completeness, we provide a proof sketch of the blackbox construction. For simplicity, we only consider the case where $g(n) = f(n) = O(\log n)$ in the proof sketch.
\begin{enumerate}
    \item Run the $\left(\frac{1}{2}, O(\log n)\right)$ low-diameter decomposition algorithm on the power graph $G^{k}$ for some $k = \Theta(1/\eps)$. This takes $O\left(\frac{\log n}{\eps}\right)$ rounds and at most $\frac{n}{2}$ vertices are unclustered.
    \item Observe that the clusters are $\Omega(1/\eps)$-hop separated in the original graph $G$. Each cluster executes a ball-growing-and-carving for $\Theta(1/\eps)$ hops and deletes the layer with the smallest number of vertices. The total number of vertices removed is upper bounded by $\frac{n/2}{\Theta(1/\eps)} = O(\eps n)$.
    \item Repeat the procedure above on the remaining unclustered vertices for $O(\log (1/\eps))$ times. At most half of the vertices remain after each run. Hence, after $O(\log (1/\eps))$ repetitions, at most $O(\eps n)$ vertices are left and can be deleted.
\end{enumerate}

\subsection{Organization}

\fullornot{
    In \cref{sect:prelim}, we present our notations and discuss some basic observations about packing and covering ILP problems.
    In \cref{sect:LDD-main}, we present our low-diameter decomposition algorithm. In \cref{sect:packing-main}, we present our algorithm for $(1-\eps)$ approximation of packing ILP problems.  In \cref{sect:covering-main}, we present our algorithm for $(1+\eps)$ approximation of covering ILP problems. In \cref{sect:conclusion}, we conclude the paper with a discussion of some open questions.
    
    In \cref{sect:concentration}, we collect all the concentration bounds needed in the paper.
    In \cref{sect:lowerbounds}, we 
    prove that $(1\pm \eps)$-approximate solutions for several combinatorial optimization problems require $\Omega\left(\frac{\log n}{\eps}\right)$ rounds to compute.
    In \cref{sect:LDD}, we review prior work on low-diameter decompositions and demonstrate a family of graphs such that the guarantee on the number of unclustered vertices does not hold with high probability for the existing low-diameter decomposition algorithms.
}{
    In \cref{sect:prelim}, we present our notations and discuss some basic observations.
    In \cref{sect:LDD-main}, we present our low-diameter decomposition algorithm. In \cref{sect:conclusion}, we conclude the paper with a discussion of some open questions.
    
    In \cref{sect:concentration}, we collect all the concentration bounds needed in the paper.
    In \cref{sect:LDD}, we review prior work on low-diameter decompositions.
    
    Our algorithm for $(1-\eps)$ approximation of packing ILP problems is in Section 4 of the full version \cite{fullversion}.
    Our algorithm for $(1+\eps)$ approximation of covering ILP problems is in Section 5 of the full version \cite{fullversion}. 
    In Appendix B of the full version \cite{fullversion}, we 
    prove that $(1\pm \eps)$-approximate solutions for several combinatorial optimization problems require $\Omega\left(\frac{\log n}{\eps}\right)$ rounds to compute.
}

\section{Preliminaries}\label{sect:prelim}

In our proofs of \cref{thm:packing,thm:covering,thm:LDD}, we may freely assume that $n$ is sufficiently large in that $n \geq n_0$ for some universal constant $n_0 > 0$, since otherwise we may solve the problem by brute force. Similarly, we may assume that $\eps$ is sufficiently small in that  $\eps \leq \eps_0$ for some universal constant $\eps_0 > 0$, since otherwise we may reset $\eps = \eps_0$. Furthermore, we may assume that $\eps = \omega\left(\frac{\log n}{n}\right)$, since otherwise $\frac{\log n}{\eps} = \Omega(n)$, so \cref{thm:packing,thm:covering,thm:LDD} become trivial.

\subsection{Notations}
We use $\dist(u,v)$ to denote the distance between $u$ and $v$. For a subset $S\subseteq V$ and a vertex $u$, we define the distance between $S$ and $v$ to be $\dist(u,S):=\min_{v \in S} \dist(u,v)$. We use $N^k(v):=\{u \in V: \dist(u,v) \leq k\}$ to denote $k$-radius neighborhood of $v$ and $N^k_{G'}(v)$ denotes the $k$-radius neighborhood of $v$ in the graph $G'$. Extending this notation to a subset of vertices $S\subseteq V$, we have $N^k(S) := \{u \in V: \dist(u,S) \leq k\}$. Note that $N^1(S) = N(S) \cup S$.

\subsection{Packing Integer Linear Programming}
We use $P: V \rightarrow \{0,1\}$ to denote a solution to a given packing problem modelled by hypergraph $H=(V,E)$. We define $W(P,S):= \sum_{v\in S} w_v P(v)$ to be the weight of any subset $S\subseteq V$ based on $P$. For any subset of vertices $S$, we use $P^{\local}_S$ to denote the optimal solution to the local packing problem restricted to the induced subgraph defined by $S$. In particular, the local packing problem is defined by setting all variables not in $S$ to zero and solving it with all constraints. Note that by setting all other variables to zero, such a solution would not violate any constraint (hyperedge), including those not fully contained in $S$. In particular, we highlight the following property.

\begin{observation}\label{lem:local_vs_global_P}
    For any subset $S \subseteq V$, let $P^*$ be a optimal solution to the packing problem. It holds that
    \[ W(P^*,S) \leq W(P^{\local}_S,S) \leq W(P^*, N^1(S))  .\]
\end{observation}
\begin{proof}
    The first inequality follows directly from the optimality of $P^{\local}_S$. 
    For the second inequality, suppose that this is not true, then we could violate the optimality of $P^*$ by using the assignment of $P^{\local}_S$ on $S$ and assigning zeros to variables in $N^1(S) \setminus S$.
\end{proof}
\subsection{Covering Integer Linear Progamming}

We use $Q: V \rightarrow \{0,1\}$ to denote a solution to a given covering problem modelled by hypergraph $H=(V,E)$. We define $W(Q,S):= \sum_{v\in S} w_v Q(v)$ to be the weight of any subset $S\subseteq V$ based on $Q$. For any subset of vertices $S$, we use $Q^{\local}_S$ to denote the optimal solution to the local covering problem restricted to the induced subgraph defined by $S$. In particular, we only consider constraints (hyperedges) where all variables are in $S$. Note that this differs with how we define local instance for packing problems. Here we discard all hyperedges that are inter-cluster. Local optimal solution for covering problem admits the following property.
\begin{observation}\label{lem:local_vs_global_C}
    For any subset $S \subseteq V$, let $Q^*$ be an optimal solution to the covering problem. It holds that
    \[ W(Q^{\local}_{S},S) \leq W(Q^*,S) \leq W(Q^*,V). \]
\end{observation}
\begin{proof}
    This follows directly from the optimality of $Q^{\local}_{S}$ and definition of $W(Q^*,\cdot)$.
\end{proof}


\section{Low-Diameter Decompositions}\label{sect:LDD-main}

In this section, we prove  \cref{thm:LDD}. We will first present an algorithm that computes  an $\left(\eps, O\left(\frac{\log^2(1/\eps) \log n}{\eps}\right)\right)$ low-diameter decomposition with high probability, and then at the end of the section we argue how the diameter bound can be improved to the ideal bound $O\left(\frac{ \log n}{\eps}\right)$. We assume that a parameter $\nn$ such that  $|V| \leq \nn \leq |V|^c$, where $c\geq 1$ is some universal constant, is initially known to all  vertices in the underlying network $G=(V,E)$. Note that $\log \nn = \Theta(\log n)$.

\subsection{The Algorithm}
We  present our algorithm that given an input graph $G = (V,E)$, computes an $\left(\eps, O\left(\frac{\log^2(1/\eps) \log n}{\eps}\right)\right)$ low-diameter decomposition with high probability. Our algorithm consists of three phases.  For the ease of presentation, we set $t :=  \ceil{\log(20/\eps)}$ and $R := \ceil{\frac{200t\ln \nn}{\eps}}$. We decompose the interval $[R+1, (t+2)R]$ into $t+1$ intervals $I_{t+1}, I_{t},\ldots, I_1$ of length $R$ by setting $I_i := [a_i, b_i] =[(t-i+2)R+1, (t-i+3)R]$.

\subsubsection{Ball-Growing-and-Carving}
The following ball-growing-and-carving procedure is a subroutine that aims to generate an isolated cluster with not-too-small radius, while not deleting too many vertices from the graph. In particular, it takes in an interval $I=[a,b]$, examines the neighborhood $N^{b}(v)$ and deletes the sparsest level in $I$ from the graph.

\begin{algorithm}[H]
\small
\caption{\textsc{Grow-and-Carve}$(I=[a, b])$ for vertex $v$\label{alg:grow_carve}}
\begin{algorithmic}[1]

\STATE Gather the topology of its $b$-radius neighborhood $N^b(v)$.
\STATE Let $S_j$ be the set of vertices of distance exactly $j$ from $v$.

\STATE Find $j^* \in I$ that minimises $|S_{j^*}|$ and delete $S_{j^*}$.
\STATE Remove $v$ and its $(j^*-1)$-radius neighborhood from the graph.

\end{algorithmic}
\end{algorithm}

We make the following distinction between ``remove'' and ``delete.'' A vertex is removed only when it belongs to some isolated cluster and hence we have already taken care of it. On the other hand, deleted vertices are permanently taken out of the graph in order to decompose the graph into \emph{non-adjacent} clusters.

\subsubsection{Phase 1}
Phase 1 proceeds in $t$ iterations. Each iteration consists of a sampling step followed by a ball-growing-and-carving. We use $G_i$ to denote the residual graph after executing the $i$th iteration and $G_0 = G$.

\begin{algorithm}[H]
\small
\caption{Phase 1 for each vertex $v$\label{alg:phase_1}}
\begin{algorithmic}[1]

\STATE Let $n_v:=|N^{4tR}(v)|$ be the number of vertices in the $(4tR)$-radius neighborhood of $v$.
\FOR{$i$ = 1 to $t$}
\STATE Sample itself to be a centre with probability $p_{v,i}=\frac{2^i \ln \nn}{n_v}$.

\IF{$v$ is a centre}
\STATE \textsc{Grow-and-Carve}$(I_i = [a_i, b_i])$.
\ENDIF

\ENDFOR

\end{algorithmic}
\end{algorithm}

Since multiple instances of \textsc{Grow-and-Carve} are running in parallel, a vertex $v$ may be simultaneously removed and deleted by different executions of \textsc{Grow-and-Carve}. Note that as long as a vertex is deleted in some execution, it is considered as deleted. 

\subsubsection{Phase 2}
Phase 2 consists of a single sampling step followed by a ball-growing-and-carving. We write $G'$ to denote the residual graph after executing Phase 2.

\begin{algorithm}[H]
\small
\caption{Phase 2 for each vertex $v$}
\begin{algorithmic}[1]

\STATE Sample itself to be a centre with probability $p_{v, t+1}=\frac{2^{t+1} \ln \nn \cdot \ln(20/\eps)}{n_v}$.

\IF{$v$ is a centre}
\STATE \textsc{Grow-and-Carve}$(I_{t+1}=[R+1, 2R])$.

\ENDIF

\RETURN

\end{algorithmic}
\end{algorithm}

\subsubsection{Phase 3}
In the last phase, we simply apply the vertex low-diameter decomposition algorithm from \cref{thm:vertex_LDD} on the residual graph with parameter $\lambda = \frac{\eps}{10}$.

\subsection{Analysis}
\begin{lemma}\label{lem:round}
    The algorithm runs in $O(t^2\cdot R) = O\left(\frac{\log^3(1/\eps) \cdot \log n}{\eps}\right)$ rounds.
\end{lemma}
\begin{proof}
    In Phase 1 and 2, each call of \textsc{Grow-and-Carve} takes at most $O(t\cdot R)$ rounds to gather information from its neighborhood. $t+1$ iterations thus take $O(t^2\cdot R)$ rounds. 
    In Phase 3, applying the algorithm from \cref{thm:vertex_LDD} takes $O\left(\frac{\ln n}{\lambda}\right) = O(R)$ rounds.
\end{proof}

\begin{lemma}\label{lem:diameter}
    When the algorithm terminates, all connected components have weak diameter $O(tR)$.
\end{lemma}
\begin{proof}
    Any vertex removed in Phase 1 and Phase 2 must be contained in a connected component of weak diameter at most $2(t+2)R$ by how we carve out the balls. Any vertex that is not deleted in Phase 3 must belong to some connected component of diameter at most $O\left(\frac{\log n}{\eps}\right) < O(tR)$, according to \cref{thm:vertex_LDD}.
\end{proof}

Let $X_{v,r}^{(i)}$ be the random variable that denotes the number of sampled centres in the $r$-radius neighborhood of any vertex $v \in V(G_{i-1})$ in the $i$th iteration. In particular,
\[ \Expect\left[X_{v,r}^{(i)}\right] = \sum_{u \in N^r_{G_{i-1}}(v)} p_{u,i}. \]
We shall show that there are not too many centres sampled around each vertex in order to control the number of deleted vertices.

\begin{lemma}\label{lem:expected_centres}
    In the $i$th iteration of Phase 1, with probability $1 - O(\nn^{-3})$, $\sum_{u \in N^{b_i}_{G_{i-1}}(v)} p_{u,i} \leq 16\ln \nn$ for all $v \in V$.
\end{lemma}
\begin{proof}
    We distinguish between two cases:
    \begin{enumerate}
        \item [$i = 1$]: Consider the $\left(b_1=(t+2)R\right)$-radius neighborhood $U:=N^{(t+2)R}(v)$ of $v$. Note that for any $u \in U$, $n_u\geq |U|$ since $4tR > 2(t+2)R$ and hence their $(4tR)$-radius neighborhood would completely cover $U$.

        We have
        \[ \sum_{u \in U} p_{u,1} = \sum_{u \in U} \frac{2\ln \nn}{n_u} \leq \sum_{u \in U} \frac{2\ln \nn}{|U|} = 2\ln \nn . \]
        \item [$i > 1$]: Assume towards contradiction that $\sum_{u \in N^{b_i}_{G_{i-1}}(v)} p_{u,i} > 16\ln \nn$. Note that $v$ is not removed in the $(i-1)$th iteration only if $X_{v,a_{i-1}}^{(i-1)}=0$. That is, there are no centres sampled in the $a_{i-1}$-radius of $v$ in the $(i-1)$th iteration. 
        On the other hand, \[ \Expect\left[X_{v,a_{i-1}}^{(i-1)}\right] = \sum_{u \in N^{a_{i-1}}_{G_{i-2}}(v)} p_{u,i-1} \geq \sum_{u \in N^{b_{i}}_{G_{i-1}}(v)} \frac{p_{u,i}}{2} > 8\ln \nn. \] Hence by a Chernoff bound,
        \[\Pr\left[X_{v,a_{i-1}}^{(i-1)} = 0\right] \leq e^{-4\ln \nn} \leq \nn^{-4} . \]
        The lemma follows by taking a union bound over all vertices, and the error probability is at most $n\cdot \nn^{-4} \leq \nn^{-3}$. \qedhere
    \end{enumerate}
\end{proof}

We remark that for the above proof to hold, we need $a_{i-1}\geq b_i$, this is the reason that the intervals $I_1, I_2, \ldots$ used in different iterations are disjoint.


\begin{lemma}\label{lem:deleted_vertices_1}
    In each iteration of Phase 1, with probability $1 - O(\nn^{-3})$, at most $\frac{\eps |V|}{4t}$ vertices are deleted.
\end{lemma}

\begin{proof}
    Consider the $i$th iteration. Due to \cref{lem:expected_centres}, each vertex $v \in V(G_{i-1})$ has $\Expect[X_{v,b_i}^{(i)}] \leq 16\ln \nn$ and only these nearby centres can cover $v$. Again by a Chernoff bound,
    \[ \Pr\left[X_{v,b_i}^{(i)} >32\ln \nn\right] \leq e^{-16\ln \nn/3} \leq \nn^{-4} .\]
    There are at most $32\ln \nn$ sampled centres in the neighborhood of $v$. In other words, $v$ can only be covered at most $32\ln \nn$ times. By a union bound over all vertices and taking into account the $O(\nn^{-3})$ error probability of \cref{lem:expected_centres}, we know that with probability $1 - O(\nn^{-3})$, all vertices are covered at most $32\ln \nn$ times.
    
    Let $F$ be the set of all centres sampled in the $i$th iteration. For any centre vertex $u \in F$, the number of deleted vertices due to $u$ is at most $\frac{|N^{b_i}(u)|}{R}$ since we are picking the most sparse layer to delete. The total number of deleted vertices is thus at most
    \[ \sum_{u \in F} \frac{|N^{b_i}(u)|}{R} \leq \frac{32\ln \nn\cdot |V|}{R} \leq \frac{\eps |V|}{4t}, \]
    where the first inequality follows from counting the total number of vertices covered in two different ways.
\end{proof}

\begin{lemma}\label{lem:deleted_vertices_2}
    In Phase 2, with probability $1 - O(\nn^{-3})$ at most $\frac{\eps |V|}{4}$ vertices are deleted.
\end{lemma}

\begin{proof}
    Following the same argument as \cref{lem:expected_centres}, we have \[\Expect\left[X^{(t+1)}_{v, 2R}\right] = \sum_{u \in N^{b_{t+1}}_{G_{t}}(v)} p_{u,t+1}\leq 16 \ln \nn \cdot \ln(20/\eps)\] for all vertices $v$, with probability $1 - O(\nn^{-3})$,  for otherwise $v$ should have been removed in Phase 1. Hence by a Chernoff bound,
    \[ \Pr\left[X_{v,2R}^{(t+1)} > 32\ln \nn \cdot \ln(20/\eps)\right] \leq e^{-16\ln \nn/3} \leq \nn^{-4}  .\]
    Each vertex $v$ can only be covered at most $32\ln \nn\cdot \ln(20/\eps)$ times. By a union bound over all vertices, this happens with probability $1 - O(\nn^{-3})$. The total number of deleted vertices is thus at most
    \[ \frac{32 \ln \nn\cdot \ln(20/\eps) \cdot |V|}{R} \leq \frac{\eps |V|}{4}. \qedhere\]
\end{proof}

Next, for the purpose of analysis, we  consider a set of bad vertices $B$ that arises from Phase 2.


\begin{definition}[Bad vertices]\label{def:bad}
The set of bad vertices $B$ is defined as follows. Consider an arbitrary decomposition of $V(G_t)$ into disjoint components $C_1, C_2, \ldots, C_k$ of weak diameter (in $G_t$) at most $R$ and size $\frac{\eps |V|}{40\ln \nn}$ and a residual part $L$ such that the $(R/2)$-radius neighborhood (in $G_t$) of each $v \in L$ contains smaller than $\frac{\eps |V|}{40\ln \nn}$ vertices in $L$. For each  component $C_i$ in the decomposition, if no vertices in $C_i$ are sampled as a centre in Phase 2, we call  $C_i$ a bad component. We define $B$ as the union of all bad components.  
\end{definition}

Note that a decomposition described in \cref{def:bad} exists, due the following greedy construction in $G_t$. For each vertex $v$, examine its $(R/2)$-radius neighborhood $N^{R/2}(v)$. If $N^{R/2}(v)$ contains at least $\frac{\eps |V|}{40\ln \nn}$ unmarked vertices, arbitrarily pick $\frac{\eps |V|}{40\ln \nn}$ unmarked vertices as one component $C_i$ and mark all of them. Repeat the procedure until no new component can be created. Each component $C_i$ is not required to be connected.

\begin{lemma}\label{lem:bad_vertices}
    The set of bad vertices $B$ has size at most $\frac{\eps |V|}{4}$  with probability $1 - O(\nn^{-4})$.
\end{lemma}

\begin{proof}
Consider the decomposition in \cref{def:bad}. 
    The number $k$ of  components $\{C_1, C_2, \ldots, C_k\}$  in the decomposition satisfies $k \leq n / \frac{\eps |V|}{40\ln \nn} = \frac{40\ln \nn}{\eps}$, as each $C_i$ has size $|C_i| = \frac{\eps |V|}{40\ln \nn}$. For each $C_i$, the probability of it being bad is at most
    \[ (1-p)^{\frac{\eps |V|}{40\ln \nn}} \leq \left(1 - \frac{40 \ln \nn}{\eps |V|}\cdot \ln(20/\eps)\right)^{\frac{\eps |V|}{40\ln \nn}}\leq e^{-\ln(20/\eps)} = \eps/20  ,\] \par
    Let $Y$ be the number of bad  components. We have that $\Expect[Y] \leq \eps k/20 \leq 2\ln \nn$.
    By a Chernoff Bound, we have 
    \[ \Pr[Y > 10\ln \nn] \leq e^{-4^2 \cdot 2\ln \nn/(2+4)} \leq \nn^{-4}  .\]
    The total number of bad vertices is therefore at most $\frac{\eps |V|}{40\ln \nn}\cdot Y \leq \frac{\eps |V|}{4}$ with high probability.
\end{proof}

Since only a small number of vertices are bad, virtually we can afford to delete all of them. While in the actual algorithm, some of them may or may not be deleted in Phase 3, we will exclude them \emph{in the analysis} so that we may apply the Chernoff Bound with bounded dependence in \cref{lem:Chernoff_BD}.  

In the following lemma, recall that  $G'$ is the residual graph after executing Phase 2, and $B$ is the set of bad vertices defined in \cref{def:bad}. The definition of a bad component ensures that if a component $C_i$ is not bad, all veritces in $C_i$ would be either be removed or deleted in Phase 2 and we will not see them in $V(G')$, so $V(G')\setminus B$ is a subset of the residual part $L$ in the decomposition of \cref{def:bad}. 




\begin{lemma}\label{lem:deleted_vertices_3}
In Phase 3, with probability $1 - O(\nn^{-2})$, at most $\frac{\eps |V|}{4}$ vertices from $V(G')\setminus B$ are deleted. 
\end{lemma}
\begin{proof}
    For each vertex $v \in V(G')\setminus B$, let $X_v$ be the indicator variable where $X_v = 1$ if $v$ is deleted in Phase 3. We have that $\Pr[X_v = 1] \leq 1 - e^{-\frac{\eps}{10}} + \nn^{-3} \leq \frac{\eps}{8}$.
    The round complexity of the execution of \cref{thm:vertex_LDD} in phase 3 is at most $\frac{4\ln \nn}{\lambda} < R/4$. 
    In other words, for any two vertices $u,v$ where $\dist(u,v)\geq R/2$, $X_u$ and $X_v$ are independent of each other. By the definition of bad vertices in \cref{def:bad}, we have that for each vertex $v\in V(G')\setminus B$, its $(R/2)$-radius neighborhood $N^{R/2}_{V(G')\setminus B}(v)$ contains at most $\frac{\eps |V|}{40 \ln \nn}$ vertices. 
    Therefore, the random variables $X_v$ have bounded dependence $d \leq \frac{\eps |V|}{40\ln \nn}$.  Let $X = \sum_{v \in V(G')\setminus B} X_v$. By applying \cref{lem:Chernoff_BD}, we have that
    \[ \Pr\left[ X > \frac{\eps |V|}{4}\right] \leq O\left(\frac{\eps |V|}{\ln \nn}\right) \cdot e^{-3 \ln \nn} \leq \nn^{-2}. \qedhere\]
\end{proof}

We are  ready to prove \cref{thm:LDD}.

\begin{proof}[Proof of \cref{thm:LDD}]
We first verify that our algorithm computes an  $\left(\eps, O\left(\frac{\log^2(1/\eps) \log n}{\eps}\right)\right)$ low-diameter decomposition  in $O\left(\frac{\log^3(1/\eps) \log n}{\eps}\right)$ rounds with high probability:
    The upper bound $\eps |V|$ on the total number of deleted vertices deleted follows from \cref{lem:deleted_vertices_1,lem:deleted_vertices_2,lem:bad_vertices,lem:deleted_vertices_3}. The round complexity bound is due to \cref{lem:round}. The  weak diameter guarantee is due to \cref{lem:diameter}. The error probability is at most $1/\nn$ from a union bound over the error of all lemmas across all iterations.
    
    The weak diameter bound $ O\left(\frac{\log^2(1/\eps) \log n}{\eps}\right)$ can be improved to the ideal \emph{strong diameter} bound $O\left(\frac{\log n}{\eps}\right)$ for free in the $\LOCAL$ model, without affecting the round complexity $O\left(\frac{\log^3(1/\eps) \log n}{\eps}\right)$. 
    We first run our algorithm with the parameter $\eps' = \frac{\eps}{2}$, and then each cluster of weak diameter  $ O\left(\frac{\log^2(1/\eps) \log n}{\eps}\right)$  locally computes an $\left(\frac{\eps}{2}, \frac{\log n}{\eps}\right)$ low-diameter decomposition by brute force. We are done by taking the union of the low-diameter decompositions over all clusters.
\end{proof}



\section{Packing Problems}\label{sect:packing-main}

In this section, we prove~\cref{thm:packing}.
Throughout the section, $H=(V,E)$ is the hypergraph representing the underlying packing ILP problem, and $P^*$ is an arbitrary fixed optimal solution to the problem that is unknown to the algorithm. We assume that a parameter $\nn$  satisfying $\max(|V|, W(P^*,V)) \leq \nn \leq \max\left(|V|,W(P^*,V)\right)^c$, where $c\geq 1$ is some universal constant, is known to all vertices in $H$.
That is, $\nn$ is  a polynomial upper bound on the number of vertices and the weight of the optimal solution. 
Recall that in this paper we only focus on the case that $W(P^*,V)$ is polynomial in $n = |V|$, so we still have $\log \nn = \Theta(\log n)$.

\subsection{Algorithm for \texorpdfstring{$(1-\eps)$}{(1-eps)}-approximate Packing}
We now present our algorithm that computes a $(1-\eps)$-approximation for any packing integer programming problem with high probability.
Let $P^*$ be an optimal solution to the problem.  
Again, we set $t :=  \ceil{\log(20/\eps)}$ and $R := \ceil{\frac{200t\ln \nn}{\eps}}$. Let $R' := R + 1$ which provides some buffer for us to apply \cref{lem:local_vs_global_P} and upper bound the loss in each carving step.
We decompose the interval $[3R'+1, 3(t+2)R']$ into $t+1$ intervals $I_{t+1}, I_{t},\ldots, I_1$ of length $3R'$, where $I_i := [a_i, b_i] =[(t-i+2)3R'+1, (t-i+3)3R']$.

\subsubsection{Preparation for Sampling}
Following the intuition of our low-diameter decomposition algorithm, we would like to sample centres with probability relating to the weight $W(P^*,V)$ of an optimal solution $P^*$, as this quantity is analogous to the number of vertices in our low-diameter decomposition algorithm.  The quantity  $W(P^*,V)$  is not known to all vertices, so we need to first compute an estimate to facilitate our sampling process.

Before everything, each vertex runs $16\ln \nn$ independent low-diameter Decompositions with $\lambda = 1/2$ from \cref{thm:vertex_LDD} in parallel. As a result, we obtain $16\ln \nn$ sets of connected components $\mathcal{C}_1, \ldots, \mathcal{C}_{16\ln \nn}$ and we denote $\mathcal{C}:= \bigcup_{i=1}^{16\ln \nn} \mathcal{C}_i$. For each connected component $C \in \mathcal{C}$, let $S_C:=N^{8tR}(C)$ be the $(8tR)$-radius neighborhood of $C$. Compute $W(P^{\local}_C,C)$ as well as $W(P^{\local}_{S_C},S_C)$.

\subsubsection{Ball-Growing-and-Carving}
The ball-growing-and-carving subroutine is different from the one for low-diameter decomposition as instead of deleting a small number vertices, we wish to delete vertices with small contribution to the optimal packing.
We emphasise that the entity executing the growing and carving is a component $C \in \mathcal{C}$ and not a vertex as compared to the low-diameter decomposition case.

\begin{algorithm}[H]\label{alg:sample_carve_P}
\small
\caption{\textsc{Grow-and-Carve-Packing}($I=[a,b]$) for $C$}
\begin{algorithmic}[1]

\STATE Gather the topology of its $\left(b-1\right)$-radius neighborhood $N^{b-1}(C)$.
\STATE Compute $P^{\local}_{N^{b-1}(C)}$.
\STATE Let $S_j$ be the set of vertices of distance exactly $j$ from $C$.
\STATE Find $j^* \in [a, b - 1]$ where $j^\ast \equiv 1 \pmod 3$,
that minimises $W(P^{\local}_{N^{b-1}(C)}, S_{j^*}\cup S_{j^*+1} \cup S_{j^*+2})$ and delete $S_{j^*+1}$. \label{step-a}
\STATE Remove $C$ and its $(j^*)$-radius neighborhood from the graph.

\end{algorithmic}
\end{algorithm}

Whenever we use the above algorithm, we will have that $a  \equiv 1 \pmod 3$ and the length $b-a+1$ of the interval $I$ is an integer multiple of $3$. Therefore, in Step~\ref{step-a} of the algorithm, the set of length-3 intervals $[j, j+2]$, over all $j \in [a, b - 1]$ with $j \equiv 1 \pmod 3$, partitions the interval $I$. 


\subsubsection{Phase 1}
Phase 1 again consists of $t$ iterations. We use $H_i$ to denote the residual graph \emph{after} executing the $i$th iteration, and we write $H_0 = H$ to denote the graph at the beginning of Phase 1.

\begin{algorithm}[H]
\small
\caption{Phase 1 for each $C \in \mathcal{C}$}
\begin{algorithmic}[1]

\FOR{$i$ = 1 to $t$}
    \STATE Sample itself to be a centre with probability $p_{C,i}=\frac{2^i W(P^{\local}_C,C)}{W(P^{\local}_{S_C},S_C)}$.
    \IF{$C$ is a centre}
        \STATE \textsc{Grow-and-Carve-Packing}($I_i = [a_i, b_i]$).
    \ENDIF
\ENDFOR
\RETURN

\end{algorithmic}
\end{algorithm}
Again, since multiple instances of \textsc{Grow-and-Carve-Packing} are running in parallel, a vertex is considered deleted as long as it is deleted by at least one execution of \textsc{Grow-and-Carve-Packing}. 
\subsubsection{Phase 2}
Phase 2 again has only one round of sampling and ball-growing-and-carving. We use $H'$ to denote the residual graph from executing Phase 2.
\begin{algorithm}[H]
\small
\caption{Phase 2 for each $C \in \mathcal{C}$}
\begin{algorithmic}[1]

\STATE Sample itself to be a centre with probability $p_{C,t+1}=\frac{2^{t+1} W(P^{\local}_C,C)\cdot \ln(20/\eps)}{W(P^{\local}_{S_C},S_C)}$.
\IF{$C$ is a centre}
    \STATE \textsc{Grow-and-Carve-Packing}($I_{t+1}= [3R'+1, 6R']$).
\ENDIF
\RETURN

\end{algorithmic}
\end{algorithm}

For the sampling process, our goal was to sieve out regions with large packing value. As the actual packing value of the hypergraph is not known to us, we estimate the packing value by having $O(\ln \nn)$ copies of the hypergraph and all of them are involved in the sampling. As a result, there is an implicit $O(\ln \nn)$ overhead in the sampling probability. This is analogous to the extra $\ln \nn$ factor in the sampling probability for our low-diameter decomposition algorithm.

\subsubsection{Phase 3: Completing the Picture}
On the residual graph, we execute the low-diameter decomposition algorithm from \cref{thm:vertex_LDD} with parameter $\lambda=\frac{\eps}{10}$. Note that at this point we essentially decompose the graph 
into connected components of weak diameter at most $O(tR)$. That is, if we look at the induced connected components arising from deleting vertices in the algorithm, they have weak diameter $O(tR)$ by how we carve out the balls and the guarantee of \cref{thm:vertex_LDD}. Each connected component then solves the packing problem locally and assign 0 to all deleted vertices.

\subsection{Analysis}

For the purpose of analysis, we shall fix an optimal solution $P^*$ and let $V^*:= \{v\in V: P^*(v) = 1\}$ and $W^* = W(P^*,V)$.
\begin{lemma}\label{lem:round_P}
    The algorithm runs in $O(t^2\cdot R) = O\left(\frac{\log^3(1/\eps) \cdot \log n}{\eps}\right)$ rounds.
\end{lemma}
\begin{proof}
    The preparation step takes $O(tR)$ rounds. Each call of \textsc{Grow-and-Carve-Packing} takes at most $O(tR)$ rounds to gather information from its neighborhood. Therefore, $t+1$ iterations  take $O(t^2\cdot R)$ rounds.
    In Phase 3, running \cref{thm:vertex_LDD} takes $O\left(\frac{\log n}{\eps}\right)$ rounds and solving local instances of packing problem for each connected component takes another $O(tR)$ rounds.
\end{proof}

\noindent Consider $H_{i-1}$ and let $X_{v,r}^{(i)}$ be the random variable that denotes the number of sampled centres that intersects with the $r$-radius neighborhood of any vertex $v \in V(H_{i-1})$ in the $i$th iteration. In particular,
\[\Expect\left[X^{(i)}_{v,r} \right] = \sum_{C\in \mathcal{C}, \; C\cap N^r_{H_{i-1}}\neq \emptyset} p_{C,i} . \]

\begin{lemma}\label{lem:expected_centres_P}
    In the $i$th iteration of Phase 1, with probability $1-O(\nn^{-3})$, for all $v \in V$ we have \[\sum_{C\in \mathcal{C}, \; C\cap N^{b_i}_{H_{i-1}}(v)\neq \emptyset} p_{C,i} \leq 32\ln \nn.\]
\end{lemma}
\begin{proof}
    We distinguish between two cases:
    \begin{enumerate}
        \item [$i = 1$]: Consider the $\left(b_1 + 3R' = (t+2)(3R')\right)$-radius neighborhood $U_v:= N^{(t+2)3R'}(v)$ of $v$. Note that for any connected component $C\in\mathcal{C}$ that intersects with  $N^{b_1}(v)$, we have $C \subseteq U_v$. Therefore, for any $j \in [1,16\ln \nn]$:
        \[ \sum_{C \in \mathcal{C}_j, \; C\cap N^{b_1}(v)\neq \emptyset} W(P^{\local}_{C},C) \leq W(P^{\local}_{U_v}, U_v) . \] 
        Moreover, $S_C$ would completely cover $U_v$ and hence
        $W(P^{\local}_{S_C},S_C) \geq W(P^{\local}_{U_v},U_v)$. We have
        \begin{align*}
            \sum_{C\in \mathcal{C}, \; C\cap N^{b_1}(v)\neq \emptyset} p_{C,1}&= \sum_{i = 1}^{16\ln \nn} \sum_{C \in \mathcal{C}_i, \; C\cap N^{b_1}(v)\neq \emptyset} \frac{2W(P^{\local}_{C},C)}{W(P^{\local}_{S_C},S_C)} \\
            &\leq \sum_{i = 1}^{16\ln \nn} \sum_{C \in \mathcal{C}_i, \; C\cap N^{b_1}(v)\neq \emptyset} \frac{2W(P^{\local}_{C},C)}{W(P^{\local}_{U_v},U_v)} \\
            &\leq \sum_{i = 1}^{16\ln \nn} \frac{2W(P^{\local}_{U_v},U_v)}{W(P^{\local}_{U_v},U_v)} \\
            &\leq 32\ln \nn  .
        \end{align*}
        \item [$i > 1$]: Assume towards contradiction that $\sum_{C\in \mathcal{C}, \; C\cap N^{b_i}_{H_{i-1}}(v)\neq \emptyset} p_{C,i} > 32\ln \nn$. Note that $v$ survives in the $(i-1)$th iteration only if $X_{v,a_{i-1}}^{(i-1)}=0$. On the other hand, $\Expect\left[X_{v,a_{i-1}}^{(i-1)}\right] \geq \sum_{C\in \mathcal{C}, \; C\cap N^{a_{i-1}}_{H_{i-2}}(v)\neq \emptyset} p_{C,i-1} \geq \sum_{C\in \mathcal{C}, \; C\cap N^{b_i}_{H_{i-1}}(v)\neq \emptyset} p_{C,i}/2 > 16\ln \nn$. Hence by a Chernoff bound,
        \[\Pr\left[X_{v,a_{i-1}}^{(i-1)} = 0\right] \leq e^{-8\ln \nn} = \nn^{-8}  . \] The lemma follows by taking a union bound over all vertices and the error probability is at most $n\cdot \nn^{-8} < \nn^{-3}$.\qedhere
    \end{enumerate}
\end{proof}

\begin{lemma}\label{lem:deleted_vertices_1_P}
    In each iteration of Phase 1, with probability $1-O(\nn^{-2})$, $W(P^*,D_1) \leq \frac{\eps W^*}{4t}$, where  $D_1$ is the set of deleted vertices in this iteration.
\end{lemma}

\begin{proof}
    Consider the $i$th iteration. Due to \cref{lem:expected_centres_P}, each vertex $v \in V(H_{i-1})$ has $\Expect\left[X_{v,b_i}^{(i)}\right] \leq 32\ln \nn$ and only these centres can cover $v$. Again by a Chernoff bound,
    \[ \Pr\left[X_{v,b_i}^{(i)} >48\ln \nn\right] \leq e^{- 32\ln \nn/(4\cdot 2.5)} \leq \nn^{-3.2}  .\]
    In other words, $v$ can only be covered at most $48\ln \nn$ times. Taking the error probability of \cref{lem:expected_centres_P} into account, the total error probability is at most $n \cdot \nn^{-3.2} + O(\nn^{-3}) < \nn^{-2}$.

    Next, we claim that $W(P^*, S_{j^*+1}) \leq W(P^{\local}_{N^{b_i-1}(C)}, S_{j^*}\cup S_{j^*+1} \cup S_{j^*+2})$, for otherwise one could violate the optimality of $P^{\local}_{N^{b_i-1}(C)}$ by using the assignment of $P^*$ on $S_{j^*+1}$ and setting variables in $S_{j^*}\cup S_{j^*+2}$ to zero. On the other hand, note that $W\left(P^{\local}_{N^{b_i-1}(C)}, N^{b_i-1}(C)\right) \leq W\left(P^*,N^{b_i}(C)\right)$ again by the optimality of $P^*$.

    Let $F$ be the set of all centres in the $i$th iteration. For any centre $C \in F$, $W(P^*, S_{j^*+1})$ is at most ${W\left(P^{\local}_{N^{b_i-1}(C)}, N^{b_i-1}(C)\right)}/R$  since we are picking the weakest layer to delete. $W(P^*,D_1)$ is thus at most
    \[ \sum_{C \in F} \frac{W(P^{\local}_{N^{b_i-1}(C)}, N^{b_i-1}(C))}{R} \leq \sum_{C \in F} \frac{W(P^*,N^{b_i}(C))}{R} \leq \frac{48\ln \nn W^*}{R} \leq \frac{\eps W^*}{4t}  , \]
    where the second inequality follows from
    \[ \sum_{C \in F} {W\left(P^*,N^{b_i}(C)\right)} = \sum_{C \in F} \sum_{v\in N^{b_i}(C)} P^*(v) \cdot w_v \leq 48 \ln \nn \sum_{v \in V} P^*(v) \cdot w_v = 48\ln \nn W^*.\qedhere\]
\end{proof}

\begin{lemma}\label{lem:deleted_vertices_2_P}
     Let $D_2$ be the set of vertices deleted in Phase 2. With probability $1 - O(\nn^{-2})$, $W(P^*,D_2) \leq \frac{\eps W^*}{4}$.
\end{lemma}

\begin{proof}
    Following the same argument of \cref{lem:expected_centres_P}, we have
    $\Expect\left[X_{v,b_{t+1}}^{(t+1)}\right] \leq 32\ln \nn\cdot \ln(20/\eps)$ with probability at least $1 - O(\nn^{-3})$ and thus
    \[\Pr\left[X_{v,b_{t+1}}^{(t+1)} > 48 \ln \nn \cdot \ln(20/\eps)\right] \leq e^{-3.2\ln \nn} = \nn^{-3.2}  . \]
    With success probability at least $1 - n\cdot \nn^{-3.2} - O(\nn^{-3}) > 1 - O(\nn^{-2})$, each vertex $v$ can only be covered at most $48\ln \nn\cdot \ln(20/\eps)$ times. Let $F$ be set the sampled centres in Phase 2, the total deleted weight is at most
    \[ \sum_{C \in F} \frac{W(P^*,N^{b_{t+1}}(C))}{R} \leq \frac{48\ln \nn\cdot\ln(20/\eps) \cdot W(P^*,V)}{R} \leq \frac{\eps W^*}{4}  , \] again because we are deleting the weakest layer. 
\end{proof}

Next, for the purpose of analysis, we shall consider a set of bad weighted vertices $B$ that arises from Phase 2. 

\begin{definition}[Weighted components]
    We say that $D = (V^D \subseteq V, w^D:V^D \rightarrow \mathbb{Z})$ is a weighted component of a weighted hypergraph $H$ if $V^D$ is a subset of $V(H)$   and $w^D(v) \leq w_v$  for each $v \in V^D$.
\end{definition}


Essentially, we are picking a component $V^D$ and then for each vertex $v\in V^D$ from the component, we take part of its weight. The function $w^D(\cdot)$ tells us how much weight we are taking for $v \in V^D$.
Alternatively, one may view a vertex $v$ with weight $w_v$ as $w_v$ copies of vertices $v^{(1)}, v^{(2)}, \ldots v^{(w_v)}$, each of weight 1. Now for each vertex $v \in V^D$, we are taking $w^D(v)$ many copies of it into $D$. 

When we delete a weighted component $D$ from some graph $H$, denoting the residual graph as $H-D$, and the meaning of deleting $D$ is that we are subtracting the weights of $H$ by weights of $D$. In other words, we have
\[ w^{H-D}_v = \begin{cases}
     w^H_v - w^D(v) & v \in V^D, \\
     w^H_v & otherwise.
\end{cases}\]
We also define $W^{H-D}(P,V) := \sum_{v\in V} w^{H-D}_v P(v)$ as the measure of the quality of a solution $P$ on the modified weighted hypergraph $H-D$.
\begin{definition}[Bad weighted vertices]\label{def:weighted_bad}
The set of bad weighted vertices $B$ is defined as follows. Consider an arbitrary decomposition of the weights of $V^* \cap V(H_t)$ into a set of weighted components $\mathcal{D} := \{D_1, D_2, \ldots\}$  and a residual part $L$ satisfying the following conditions.
\begin{itemize}
    \item Each $D \in \mathcal{D}$ has weak diameter at most $R$ in $H_t$ and satisfies $\sum_{v\in V^D} w^D(v) = \ceil{\frac{\eps W^*}{40\ln \nn}}$.
    \item For each $v \in L$, its  $(R/2)$-radius neighborhood in $H_t$ satisfies that 
\[ \sum_{u \in V^*\cap N^{R/2}(v)} w^L(u) \leq \ceil{\frac{\eps W^*}{40\ln \nn}} .\]
    \item The weighted components $\mathcal{D}$ and $L$ decompose $V^* \cap V(H_t)$ in the sense that 
    \[  w^{L}_v + \sum_{D\in \mathcal{D} \; : \; v\in V^D} w^D(v)  = w_v \]
    for each $v \in V^* \cap V(H_t)$. That is, the weights add up correctly.
\end{itemize}
For each component $D \in \mathcal{D}$, if no vertices in $D$ are sampled as a centre in Phase 2, we call it a bad component. The set $B$ is the union of all bad components.
\end{definition}

Similar to \cref{def:bad}, a decomposition described in \cref{def:weighted_bad} exists, due to the following greedy construction in $H_t$. For each vertex $v$, examine its $(R/2)$-radius neighborhood $N^{R/2}(v)$ in $H_t$. If  greater than $\ceil{\frac{\eps W^*}{40\ln \nn}}$ weight of $V^\ast$ were there, arbitrarily pick $\ceil{\frac{\eps W^*}{40\ln \nn}}$ units of weight to form one weighted component $D_i$ and eliminate them from $H_t$. Repeat the procedure until no new weighted component can be created. Note that each weighted component $D \in \mathcal{D}$ is not required to be connected. 

\begin{lemma}\label{lem:bad_vertices_P}
    $\sum_{D \in \mathcal{D} \; : \; \text{$D$ is bad}}  \sum_{v \in V^D}P^*(v)\cdot w^D(v) \leq \frac{\eps W^*}{4}$ with probability $1 - O(\nn^{-3})$.
\end{lemma}

\begin{proof}
    Consider the following alternative sampling process $\mathsf{SAMP}'$: For all $D \in \mathcal{D}$ and $C \in \mathcal{C}$, if a vertex $v\in D\cap C$, $v$ marks itself with probability $p_{D,C,v} = \frac{40 w^D(v)}{ \eps W(P^{\local}_{S_C},S_C)}\cdot \ln(20/\eps)$. This happens independently for all $C$ and $D$, i.e., it is possible that $v$ is marked for some pair $(C,D)$ but not marked for another pair $(C',D)$.
    For each connected component $C \in \mathcal{C}$, it is sampled if any of its vertices is marked. By union bound, the probability of sampling $C$ in $\mathsf{SAMP}'$ is at most
    \begin{align*}
        \sum_{D \in \mathcal{D}, \; v \in C\cap D} p_{D,C,v} \leq \frac{40 W(P^*,C)\cdot \ln(20/\eps)}{\eps W(P^{\local}_{S_C},S_C)}\leq \frac{40 W(P^{\local}_C,C)\cdot \ln(20/\eps)}{\eps W(P^{\local}_{S_C},S_C)} = p_{C,t+1},
    \end{align*}
    since $2^{t+1} \geq \frac{40}{\eps}$.
    Hence, there is a coupling between $\mathsf{SAMP}'$ and the actual sampling process $\mathsf{SAMP}$, in the sense that if component $C$ is sampled in $\mathsf{SAMP}$, it is also sampled in $\mathsf{SAMP}'$. 

    Let $k$ be the number of components $D$ in $\mathcal{D}$ and we have $k \leq W^* / \frac{\eps W^*}{40\ln \nn} = \frac{40\ln \nn}{\eps}$. Note that for any such $D$, it is bad only if all of its vertices are not marked. This happens with probability at most
    \[ \prod_{C \in \mathcal{C}, \; v \in C\cap D} (1-p_{D,C,v}) \leq e^{-\ln(20/\eps)} = \eps/20  .\] \par
    This is because $\sum_{v \in D} w^D(v) \geq \frac{\eps W^*}{40 \ln \nn}$ by definition of the weighted decomposition. Moreover, in expectation any vertex $v$ should be contained in greater than $\left(e^{-0.5} - \nn^{-3}\right)16 \ln \nn \geq 8\ln \nn$ connected components $C \in \mathcal{C}$. by a Chernoff bound, with probability at least $1 - e^{-(7/8)^2 8\ln \nn/2} \geq 1 - \nn^{-3}$, a vertex $v$ is contained in at least $\ln \nn$ connected components $C \in \mathcal{C}$.
    So that we have
    \[ \sum_{C \in \mathcal{C}, \; v \in C\cap D} p_{D,C,v} \geq \ln \nn \cdot \frac{40}{\eps W^*} \cdot \frac{\eps W^*}{40 \ln \nn}\cdot \ln(20/\eps) = \ln(20/\eps)  . \]
    Let $Y$ be the number of bad weighted components. We have that $\Expect[Y] \leq \eps t/20 \leq 2\ln \nn$.
    by a Chernoff bound, we have that $\Pr[Y > 10\ln \nn] \leq \nn^{-4}$.
    The total weight of bad weighted components is therefore at most $\frac{\eps W^*}{40\ln \nn}\cdot Y \leq \frac{\eps W^*}{4}$ with   probability $1 - O(\nn^{-3})$.
\end{proof}

Since the bad components contribute only a small weight, it would not hurt us even if we ignore them and treat them as deleted. Hence, in the remaining parts of the analysis, we will exclude the weights in $B$. Note that since $B$ consists of weighted components, a vertex may be partially in $B$ in that only a fraction of its weight is in $B$. We will show that by excluding $B$ in the analysis, the remaining graph satisfies the condition to apply the weighted case of Chernoff Bound with bounded dependence from \cref{lem:Chernoff_BD_weighted}.

In the following lemma, recall that $H'$ is the residual graph after executing Phase 2, and $B$ is a collection of weighted vertices defined in \cref{def:weighted_bad}. We will be examining the weighted graph $H' - B$. In particular, if a weighted component $D \in \mathcal{D}$ is not bad, then it is completely removed or deleted in Phase 2. Hence $H' - B$ is indeed a subset of the residual part $L$ from \cref{def:weighted_bad}.


\begin{lemma}\label{lem:deleted_vertices_3_P}
 Let $D_3$ be the set of vertices deleted in Phase 3. With probability $1-O(\nn^{-2})$, $W^L(P^*,D_3) \leq \frac{\eps W^*}{4}$.
\end{lemma}

\begin{proof}
    By the definition of weighted component decomposition, we have that for each vertex $v\in V(H')$, $\sum_{u \in V^*\cap N^{R/2}(v)} w^L(u) \leq \ceil{\frac{\eps W^*}{40\ln \nn}}$. For any vertex $v \in V^*$, let $X_v$ be the indicator variable where $X_v = 1$ if $v$ is deleted in Phase 3. The round complexity of the execution of \cref{thm:vertex_LDD} in phase 3 is at most $4\ln \nn/\lambda < R/4$. In other words, for any two vertices $u,v$ where $\dist(u,v)\geq R/2$, $X_u$ and $X_v$ are independent of each other. 
    Hence we have that $\Pr[X_v = 1] \leq 1-e^{-\frac{\eps}{10}} + \nn^{-3} \leq \eps/8$ and they have bounded dependence $d \leq \frac{\eps W^*}{40\ln \nn}$. By applying \cref{lem:Chernoff_BD_weighted}, we have that $X = \sum_v w_v X_v $ satisfies: 
    \[ \Pr\left[ X > \frac{\eps W^*}{4}\right] \leq O\left(\frac{\eps W^*}{\ln \nn}\right) \cdot e^{-3 \ln \nn} < \nn^{-2}. \qedhere\]
\end{proof}


We are ready to prove \cref{thm:packing}.

\begin{proof}[Proof of \cref{thm:packing}]
We first show that our algorithm computes an $(1 - \eps)$-approximate solution to the packing problem. First of all, $H$ is decomposed into non-adjacent components and all deleted vertices are assigned 0. We verify that all  constraints are satisfied by check each hyperedge $e \in E$:
\begin{itemize}
    \item $e$ is fully contained in some component $S_i$: $e$ is satisfied by the local solution handling $S_i$.
    \item $e$ is adjacent to more than one components: Such $e$ does not exist as the components induced from the deletion of vertices are non-adjacent.
    \item $e$ is adjacent to some component $S_i$ and some deleted vertices: By how we formulate the local instance of a packing problem and the fact that deleted vertices are assigned 0, $e$ is satisfied by the local solution handling $S_i$.
\end{itemize}
Hence, the solution we computed does not violate any constraint. Next, let $D$ be the set of deleted vertices throughout the algorithm. By \cref{lem:deleted_vertices_1_P,lem:deleted_vertices_2_P,lem:bad_vertices_P,lem:deleted_vertices_3_P}, we have $W(P^*,D) \leq \eps W(P^*,V)$. By the optimality of solutions for each connected component, the solution has total weight at least $W(P^*, V \setminus D) = W(P^*,V) - W(P^*, D) \geq (1-\eps) W(P^*,V)$, as required. The error probability is at most $1/\nn$ from a union bound over the error of all lemmas across all iterations.
Finally, the round complexity bound is due to \cref{lem:round_P}.
\end{proof}



\paragraph{An Alternative Approach}
Below we describe an alternative approach to proving \cref{thm:packing}.\footnote{We thank an anonymous reviewer for suggesting this approach.}
Suppose that we already have a weighted low-diameter decomposition algorithm $\mathcal{A}$, i.e., given a weighted graph $G$, $\mathcal{A}$  computes a low-diameter decomposition of $G$ where the weight of the deleted vertices is at most an $\eps$ fraction of the total weight, with high probability. Such an algorithm $\mathcal{A}$ can be obtained by extending our low-diameter decomposition  algorithm of~\cref{thm:LDD} to the weighted case by incorporating some techniques developed in this section. 

Similar to the preparation step, in the alternative approach, the algorithm starts by running $t = O(\eps^{-2} \log \nn)$ low-diameter decomposition algorithms from~\cite{EN16} in parallel and computes the corresponding packing solutions $P_i : V \rightarrow \{0,1\}$ for the $i$th run of the decomposition. Let $w(P_i)$ be the objective value of the solution $P_i$. Next, we assign a new weight  $w'(v):= w(v) \cdot |\{i: P_i(v) = 1\}|$ to each vertex $v \in V$, apply the weighted low-diameter decomposition algorithm $\mathcal{A}$, and compute a packing solution $P'$ in the decomposed graph. 

Let $P^*$ be an optimal solution. The initial low-diameter decompositions of~\cite{EN16} give an expectation guarantee: $\Expect[w(P_i)] \geq (1-\eps)w(P^*)$. By a Chernoff bound, $\sum_i w(P_i)$ should be concentrated around $t(1-\eps)^2  w(P^*)$ with high probability, and the weighted low-diameter decomposition should have clustered a total weight of at least $t(1 - \eps)^3  w(P^*)$. Let $P_i'$ be the partial packing solution restricting to clustered vertices in $P_i$, i.e., $P_i'(v) = 1$ if and only if $v$ is clustered and $P_i(v) = 1$. By an averaging argument, there must be an index $i$ where $w(P_i') \geq (1 - \eps)^3 w(P^*)$. The correctness of the overall algorithm then follows from $w(P') \geq w(P_i')$.



\section{Covering Problems}\label{sect:covering-main}
In this section, we prove \cref{thm:covering}. Similar to the previous section, we let $H=(V,E)$ be the hypergraph corresponding to the underlying covering ILP problem, and we let  $Q^*$ be any fixed optimal solution to the problem, which is unknown to the algorithm.
We assume that there is a parameter $\nn$  such that $\max(|V|, W(Q^*,V)) \leq \nn \leq \max\left(|V|,W(Q^*,V)\right)^c$, where $c\geq 1$ is some universal constant, is known to all vertices. As we only focus on the case $W(Q^*,V)$ is polynomial in $n = |V|$, we also have $\log \nn = \Theta(\log n)$.

\subsection{Algorithm for \texorpdfstring{$(1+\eps)$}{(1+eps)}-approximate Covering}
We now present our algorithm that computes a $(1+\eps)$-approximation for any covering ILP problem with high probability. 
We set $t := \ceil{\log\ln n+\log(1/\eps)+8}$, $R := \ceil{\frac{200t\ln \nn}{\eps}}$.
We decompose the interval $[2R+1, 2(t+1)R]$ into $t$ intervals $I_{t}, I_{t-1},\ldots, I_1$ of length $2R$, where $I_i := [a_i, b_i] =[(t-i+1)2R+1, (t-i+2)(2R)]$.
Note that here each interval has $2R$ layers of vertices and we shall view this as $R$ layers of disjoint hyperedges.

\subsubsection{Preparation for Sampling}
Similar to the packing case, we need to compute an estimate of the optimal covering value to facilitate our sampling process.

Before everything, each vertex runs $16\ln \nn$ independent executions of \cref{thm:LDD_covering} with $\lambda =(\ln 20/21)$ in parallel. As a result, we obtain $16\ln \nn$ sets of connected components $\mathcal{C}_1, \ldots, \mathcal{C}_{16\ln \nn}$ and we denote $\mathcal{C}:= \bigcup_{i=1}^{16\ln \nn} \mathcal{C}_i$. For each component $C \in \mathcal{C}$, we compute $W(Q^{\local}_{C},C)$. Let $S_C:=N^{8tR}(C)$ be the $(8tR)$-radius neighborhood of $C$, also compute $W(Q^{\local}_{S_C},S_C)$. 

\subsubsection{Ball-Growing-and-Carving}
The ball-growing-and-carving subroutine differs from the one for packing as we are now deleting hyperedges (i.e., constraints) instead of vertices. We can afford to delete vertices in the packing case because assigning zero to deleted vertices would not violate any constraint. However this is not the case for covering. If we were to delete vertices, analogously we need to assign one to all deleted vertices to ensure no constraints are violated, which would incur too much loss and hurt our approximation ratio. Hence, we instead try to satisfy some constraints that does not cost too much, and delete them to partition the graph into non-adjacent parts.

\begin{algorithm}[H]\label{alg:sample_carve_C}
\small
\caption{\textsc{Grow-and-Carve-Covering}($I=[a,b]$) for $C$}
\begin{algorithmic}[1]

\STATE Gather the topology of its $b$-radius neighborhood $N^{b}(C)$.
\STATE Compute $Q^{\local}_{N^b(C)}$. 
\STATE Let $S_j$ be the set of vertices of distance exactly $j$ from $C$.
\STATE Find odd integer $j^* \in [a, b]$ that minimises $W(Q^{\local}_{N^b(C)}, S_{j^*} \cup S_{j^*+1})$.
\STATE Fix the assignment of $Q^{\local}_{N^b(C)}$ on $S_{j^*}$ and $S_{j^*+1}$, delete hyperedges between them (note that these hyperedges are now satisfied).
\STATE Remove $C$ and its $(j^*)$-radius neighborhood from the graph.

\end{algorithmic}
\end{algorithm}
Note that by ``fixing assignment'', it means that if a variable is assigned one in some execution, it is permanently assigned one while the rest of the variables are still free to be assigned. 

\subsubsection{Phase 1}
Phase 1 consists of $t$ iterations. We use $H_i$ to denote the residual graph \emph{after} executing the $i$th iteration and $H_0 = H$.

\begin{algorithm}[H]
\small
\caption{Phase 1 for each $C \in \mathcal{C}$}
\begin{algorithmic}[1]

\FOR{$i$ = 1 to $t$}
    \STATE Sample itself to be a centre with probability $p_{C,i}=\frac{2^i W(Q^{\local}_{C},C)}{W(Q^{\local}_{S_C},S_C)}$.
    \IF{$C$ is a centre}
        \STATE \textsc{Grow-and-Carve-Covering}($I_i = [a_i, b_i]$).
    \ENDIF
\ENDFOR
\RETURN

\end{algorithmic}
\end{algorithm}

Similar to the packing case, while running in parallel, an hyperedge is deleted as long as it is deleted by some cluster.

\subsubsection{Phase 2: Completing the Picture}
On the residual graph $H_t$, we execute the covering algorithm by combining  \cref{thm:LDD_covering,thm:simple_covering} with parameter $\lambda = \ln \frac{\eps + 5}{5}$. Also, for vertices removed from Phase 1, they form induced components with weak diameter at most $(t+1)R$ that are non-adjacent to any other components. Each of these isolated components solve their local covering instance in parallel.

\subsection{Analysis}

For the purpose of analysis, we shall fix an optimal solution $Q^*$ and let $V^*:= \{v\in V: Q^*(v) = 1\}$. Also, let $W^* := W(Q^*,V)$.
\begin{lemma}\label{lem:round_C}
    The algorithm runs in $O(t^2\cdot R) = O\left(\frac{\left(\log\log n + \log(1/\eps)\right)^3 \cdot \log n}{\eps}\right)$ rounds.
\end{lemma}
\begin{proof}
    The preparation step takes $O(tR)$ rounds. Each call of \textsc{Grow-and-Carve-Covering} takes at most $O(tR)$ rounds to gather information from its neighborhood. $t$ iterations thus take $O(t^2\cdot R)$ rounds. 
    In Phase 2, running \cref{thm:LDD_covering} takes $O\left(\frac{\ln n}{\lambda}\right) = O\left(\frac{\ln n}{\eps}\right)$ rounds and solving local instances of covering problem for each connected component takes another $O(tR)$ rounds.
\end{proof}

\noindent let $X_{v,r}^{(i)}$ be the random variable that denotes the number of sampled centres that intersect with the $r$-radius neighborhood of any vertex $v \in V(H_{i-1})$ in the $i$th iteration. In particular,
\[\Expect\left[X^{(i)}_{v,r}  \right]= \sum_{C\in \mathcal{C}, \; C\cap N^r_{H_{i-1}}(v)\neq \emptyset} p_{C,i} . \]

\begin{lemma}\label{lem:expected_centres_C} 
    In the $i$th iteration of Phase 1, with probability $1-O(\nn^{-1.17})$, for all $v \in V$, we have \[\sum_{C\in \mathcal{C}, \; C\cap N^{b_i}_{H_{i-1}}(v)\neq \emptyset} p_{C,i} \leq 64\ln \nn.\] 
\end{lemma}
\begin{proof}
    We distinguish between two cases:
    \begin{enumerate}
        \item [$i = 1$]: Consider the $\left(b_1 + 2R = (t+2)(2R)\right)$-radius neighborhood $U_v := N^{(t+2)(2R)}(v)$ of $v$. Note that for each connected component $C\in\mathcal{C}$ that intersects with $N^{b_1}(v)$ the $b_1$-radius neighborhood of $v$, $S_C$ would completely cover $U_v$ and hence
        $W(Q^{\local}_{S_C},S_C) \geq W(Q^{\local}_{U_v},U_v)$. Moreover, we have $C \subset U_v$. Now according to \cref{thm:simple_covering}, for any $j \in [1,16\ln \nn]$,
        \[\sum_{C \in \mathcal{C}_j, \; C\cap N^{b_1}(v)\neq \emptyset} W(Q^{\local}_{C},C) \leq \sum_{C \in \mathcal{C}_j, \; C\cap N^{b_1}(v)\neq \emptyset} W(Q^{\local}_{U_v},C) \leq \sum_{u\in U_v} Y_u^{(j)} \cdot Q^{\local}_{U_v}(u) \cdot w_u, \]
        where $Y_u^{(j)}$ counts the number of subsets $u$ belongs to in $\mathcal{C}_j$ and is dominated by $Z_j + \nn^{-2}$ for $Z_j\sim \geom(\frac{20}{21})$. Note that $\Expect\left[Z_j \right] \leq 1.05$. Therefore,
        \[\sum_{j=1}^{16\ln \nn} \sum_{C \in \mathcal{C}_j, \; C\cap N^{b_1}(v)\neq \emptyset} W(Q^{\local}_{C},C) \leq \sum_{u\in U_v}\left(\sum_{j=1}^{16\ln \nn} Y_u^{(j)}\right) \cdot Q^{\local}_{U_v}(u) \cdot w_u. \]
        By \cref{lem:Chernoff_Geometric}, we have that for each vertex $u \in V$: 
        \begin{align*} \Pr\left[\sum_{j=1}^{16\ln \nn} Y_u^{(j)} > 32\ln \nn\right] &\leq\Pr\left[\sum_{j=1}^{16\ln \nn} Z_j > 32\ln \nn - \frac{16\ln\nn}{\nn^2}\right]\\
        &\leq\Pr\left[\sum_{j=1}^{16\ln \nn} Z_j > (1.05) \cdot 16\ln \nn + (0.9)\cdot 16 \ln\nn \right]\\
        &\leq e^{- (0.9 p^2)16 \ln \nn / 6} < \nn^{-2.17} ,
        \end{align*}
        where $p = \frac{20}{21}$. Therefore,
        \begin{align*}
        \sum_{C\in \mathcal{C}, \; C\cap N^{b_1}(v)\neq \emptyset} p_{C,1} &= \sum_{j = 1}^{16\ln \nn} \sum_{C \in \mathcal{C}_j, \; C\cap N^{b_1}(v)\neq \emptyset} \frac{2W(Q^{\local}_{C},C)}{W(Q^{\local}_{S_C},S_C)} \\
        &\leq \sum_{j = 1}^{16\ln \nn}\sum_{C \in \mathcal{C}_j, \; C\cap N^{b_1}(v)\neq \emptyset} \frac{2W(Q^{\local}_{C},C)}{W(Q^{\local}_{U_v},U_v)} \\
        &\leq \frac{64 \ln \nn\cdot W(Q^{\local}_{U_v},U_v)}{W(Q^{\local}_{U_v},U_v)} \\
        &\leq 64\ln \nn  .
        \end{align*}
        \item [$i > 1$]: Assume towards contradiction that $\sum_{C\in \mathcal{C}, \; C\cap N^{b_i}_{H_{i-1}}(v)\neq \emptyset} p_{C,i} > 64\ln \nn$. Note that $v$ survives in the $(i-1)$th iteration only if $X_{v,a_{i-1}}^{(i-1)}=0$. On the other hand, \[\Expect\left[X_{v,a_{i-1}}^{(i-1)}\right] \geq \sum_{C\in \mathcal{C}, \; C\cap N^{a_{i-1}}_{H_{i-2}}(v) \neq \emptyset} p_{C,i-1} \geq \sum_{C\in \mathcal{C}, \; C\cap N^{b_i}_{H_{i-1}}(v)\neq \emptyset} \frac{p_{C,i}}{2} > 32\ln \nn.\] By a Chernoff bound,
        \[ \Pr\left[X_{v,a_{i-1}}^{(i-1)}=0\right] \leq e^{-16 \ln \nn} = \nn^{-16}  . \]
        The lemma follows by taking a union bound over all vertices, where the error probability is at most $n\cdot \nn^{-2.17} \leq \nn^{-1.17}$.\qedhere
    \end{enumerate}
\end{proof}


\begin{lemma}\label{lem:deleted_vertices_1_C}
    In each iteration of Phase 1, with probability $1 - O(\nn^{-1.17})$, the total weight from fixing the assignment in the \textsc{Grow-and-Carve-Covering} subroutine is at most $\frac{\eps W(Q^*,V)}{2t}$.
\end{lemma}

\begin{proof}
    Consider the $i$th iteration. Due to \cref{lem:expected_centres_C}, with error probability at most $O(\nn^{-1.17})$, each vertex $v \in V(H_{i-1})$ has $\Expect\left[X_{v,b_i}^{(i)}\right] \leq 64\ln \nn$ and only these centres can cover $v$. Again by a Chernoff bound,
    \[ \Pr\left[X_{v,b_i}^{(i)} > (1+0.5) 64 \ln \nn\right] \leq e^{- 64 \ln \nn/ 10} = \nn^{-6.4}  . \]
    In other words, with probability at least $1 - O(\nn^{-1.17}) - n\cdot \nn^{-6.4} \geq 1 - O(\nn^{-1.17})$, all vertices $v$ can only be covered at most $(1.5)\cdot 64\ln \nn = 96\ln \nn$ times. Also, we recall that \[W\left(Q^{\local}_{N^{b_i}(C)}, N^{b_i}(C)\right) \leq W(Q^*, N^{b_i}(C)).\]
    Let $F$ be the set of all sampled centres in the $i$th iteration. For any centre $C \in F$, $W\left(Q^{\local}_{N^{b_i}(C)}, S_{j^*}\cup S_{j^*+1}\right)$ is at most $
    W\left(Q^{\local}_{N^{b_i}(C)}, N^{b_i}(C)\right)/R$ since we are picking the weakest layer to delete. The fixed assignment weight is at most 
    \[ \sum_{C \in F} \frac{W\left(Q^{\local}_{N^{b_i}(C)}, N^{b_i}(C)\right)}{R} \leq \sum_{C \in F} \frac{W\left(Q^*,N^{b_i}(C)\right)}{R} \leq \frac{96\ln \nn W(Q^*,V)}{R} \leq \frac{\eps W(Q^*,V)}{2t}.\qedhere \]
\end{proof}


\begin{lemma}\label{lem:sparsity_covering}
    With probability $1-O(\nn^{-4})$, we have $W(Q^*, N^{2R}(v))\leq \frac{\eps W(Q^*,V)}{250\ln \nn}$ for all vertices $v$ in $H_t$.
\end{lemma}

\begin{proof}
    For any vertex $v \in V(H_t)$, assume towards contradiction that $W(Q^*, N^{2R}(v)) > \frac{\eps W(Q^*,V)}{250\ln \nn}$. It survives the last iteration only if no component $C \in \mathcal{C}$ intersecting $N^{2R}(v)$ are sampled as centres. For any $j \in [1,16\ln \nn]$, Note that $\bigcup_{C\in \mathcal{C}_j, C \cap N^{2R}(v) \neq \emptyset}$ completely covers $N^{2R}(v)$. Hence,
    \[ \sum_{C\in \mathcal{C}_j, \; C \cap N^{2R}(v) \neq \emptyset} W(Q^{\local}_{C}, C) \geq W(Q^*,N^{2R}(v)) > \frac{\eps W^*}{250\ln \nn} . \]

    As such, the probability that none of them being sampled is at most
    \[\prod_{C\in \mathcal{C}, \; C \cap N^{2R}(v) \neq \emptyset} \left(1 - \frac{256 \ln \nn}{\eps}\frac{W(Q^{\local}_{C},C)}{ W(Q^{\local}_{S_C},S_C)}\right) \leq e^{-\big(16\ln \nn \cdot \frac{256\ln \nn \cdot W(Q^*,N^{2R}(v))}{\eps W^*} \big) } \leq  e^{-16\ln \nn} \leq \nn^{-16}  .\]
    The rest follows from a union bound over all vertices $v$.
\end{proof}

\begin{lemma}\label{lem:deleted_vertices_2_C}
In Phase 2, the computed solution has weight at most $(1+\frac{\eps}{2})W(Q^*,V)$ with probability $1-O(\nn^{-3})$.
\end{lemma}

\begin{proof}
    Let $U$ be the set of removed vertices from Phase 1. Noticing that components formed by the $U$ have weak diameter $O(tR)$ and are fully isolated, the local solution for $U$ has weight 
    \[W(Q^\local_U, U) \leq W(Q^*, U) .\]
    Next, by \cref{lem:sparsity_covering}, we have that with error probability at most $O(\nn^{-4})$, $W(Q^*,N^{2R}(v)) \leq \frac{\eps W^*}{250\ln \nn}$ for each vertex $v \in V(H_t)$. By \cref{thm:simple_covering}, the weight of the solution is at most 
        \begin{align*}
            \sum_{v \in V(H_t)} X_v \cdot Q^*(v) \cdot w_v  ,
        \end{align*}
    where $X_v$ is the number of clusters that contains $v$ as defined in \cref{thm:LDD_covering} and $X_v$ is dominated by $Z_v + \nn^{-2}$ for $Z_v$ being geometric random variable with $\Expect[Z_v] \leq 1 + \eps/5$
    by our choice of $\lambda$.
    Note that we can apply \cref{lem:geometric_weighted} and argue that
    \begin{align*}
        &\Pr\left[\sum_{v\in V^*\cap V(H_t)} X_v \cdot w_v > \left(1+\frac{\eps}{2}\right) \cdot W(Q^*,V(H_t)) \right] \\
        &\Pr\left[\sum_{v\in V^*\cap V(H_t)} (Z_v + \nn^{-2}) \cdot w_v > \left(1+\frac{\eps}{5} + \frac{\eps}{4} + \frac{\eps}{20}\right) \cdot W(Q^*,V(H_t)) \right] \\
        &\Pr\left[\sum_{v\in V^*\cap V(H_t)} Z_v \cdot w_v > \left(\left(1+\frac{\eps}{5}\right) + \frac{\eps}{4} \right) \cdot W(Q^*,V(H_t)) \right] \\
        &\leq O\left(\frac{W^*}{\log n}\right) e^{- \frac{p^2 \eps W^*}{48d}} \leq O\left(\frac{W^*}{\log n}\right) e^{-3.6\ln \nn} \leq \nn^{-3},
    \end{align*}
    where $p = \frac{5}{5+\eps} > \frac{5}{6}$.
    So the total weight of the computed solution is at most
    \[W(Q^*,U) + \left(1 + \frac{\eps}{2}\right) W(Q^*,V(H_t)) \leq  \left(1 + \frac{\eps}{2}\right) W(Q^*,V(H_t)\cup U) \leq  \left(1 + \frac{\eps}{2}\right) W(Q^*,V). \qedhere\]
\end{proof}

We are ready to prove \cref{thm:covering}.

\begin{proof}[Proof of \cref{thm:covering}]
    We first show that our algorithm computes an $(1 + \eps)$-approximate solution to the covering problem. In Phase 1, a constraint (hyperedge) is deleted only if it is already satisfied and we pay $(\eps/2) W^*$ for deleting them from \cref{lem:deleted_vertices_1_C}. In Phase 2, following from \cref{thm:simple_covering}, we output a valid solution to the covering problem and the weight of the solution is $(1 + \eps / 2)W^*$ from \cref{lem:deleted_vertices_2_C}. Hence the final solution that we obtain has weight $(1 + \eps)W^*$, as required. The error probability is at most $1/\nn$ from a union bound over the error of all lemmas across all iterations.
    Finally, the round complexity follows from \cref{lem:round_C}.
\end{proof}



\section{Conclusions and Open Questions}\label{sect:conclusion}

In this paper, we showed that low-diameter decompositions and  $(1\pm \eps)$-approximate solutions for general packing and covering ILP problems can be computed in $\tilde{O}\left(\frac{\log n }{\eps}\right)$ rounds \emph{with high probability} in the $\LOCAL$ model. It still remains an open question whether the \emph{ideal} bound $O\left(\frac{\log n}{\eps}\right)$ can be achieved for these problems. An even more challenging question is to determine whether the bound $O\left(\frac{\log n}{\eps}\right)$ can be achieved \emph{deterministically}.

For low-diameter decompositions, a natural research question is to extend our algorithm to the $\CONGEST$ model. A straightforward extension will add an $O(\log n)$ factor to the round complexity, due to the fact that in each iteration, each vertex is involved in up to $O(\log n)$ ball-growing-and-carving. Furthermore, due to the overlap between these $O(\log n)$ balls, we only obtain a \emph{weak-diameter} decomposition. The ultimate goal of this research direction will be to design an $O\left(\frac{\log n}{\eps}\right)$-round algorithm in the $\CONGEST$ model that constructs a low-diameter decomposition such that each cluster has \emph{strong diameter} $O\left(\frac{\log n}{\eps}\right)$ and the bound $\eps |V(G)|$ on the number of unclustered vertices holds \emph{with high probability}.

As already discussed in \cref{sect:method}, it is an open question whether a spanner of stretch $2k - 1$ and size $O\left(n^{1+1/k}\right)$ can be computed in $O(k)$ rounds \emph{with high probability}. This open question was stated in~\cite{forsterOPODIS2021}. The existing construction of such a spanner with  \emph{expected} size $O\left(n^{1+1/k}\right)$ is based on a variant of low-diameter decomposition where the diameter of each cluster is $2k - 2$. Here $k$ is typically a small integer. It will be interesting to see if our techniques can be applied to such a variant of low-diameter decomposition to resolve the open question. 


Unlike our algorithms which apply to \emph{all} covering and packing ILPs, the $\Omega\left(\frac{\log n}{\eps}\right)$ lower bounds \fullornot{in \cref{sect:lowerbounds}}{(see Appendix B of the full version \cite{fullversion})} \emph{do not} apply to \emph{all} non-trivial covering and packing integer linear programs. In particular, the specific lower bound proof presented in this paper fails to give a non-trivial lower bound for $(1-\eps)$-approximate \emph{maximum matching}, as it is known~\cite{flaxman2007maximum} that all high-girth regular graphs admit a large matching which includes all but an exponentially small fraction $O\paren{(d-1)^{\frac{g}{2}}}$ of the vertices, where $g$ is the girth and $d$ is the degree.
It is still an intriguing open question to determine the optimal asymptotic round complexity of distributed $(1-\eps)$-approximate maximum matching, in both $\LOCAL$ and $\CONGEST$. It is known~\cite{BCGS17} that $(1-\eps)$-approximate maximum matching can be solved in $2^{O(1/\eps)} \cdot O\left(\frac{\log \Delta}{\log \log \Delta}\right)$ rounds in the $\CONGEST$ model, where $\Delta < n$ is the maximum degree of the graph, so at least we know that the bound $O\left(\frac{\log n}{\eps}\right)$ is not tight for some range of $\eps$.

\bibliographystyle{alpha}
\bibliography{references}

\appendix
\section{Concentration Bounds}\label{sect:concentration}



\begin{lemma}[Chernoff bound]
    Let $X_1, \ldots, X_n$ be independent 0-1 random variables, $X = \sum_i X_i$ be their sum and $\mu = \Expect[X]$. Then
\begin{align*}
    \Pr[X > (1+\delta) \mu] &\leq e^{-\delta^2\mu / (2+\delta)} && \text{for} \ \ \delta \geq 0,\\
    \Pr[X < (1-\delta) \mu] &\leq e^{-\delta^2\mu / 2} && \text{for} \ \  0 \leq \delta \leq 1.
\end{align*}
\end{lemma}

\fullornot{
    We use the following definition for geometric distribution. Let $X \sim \geom(p)$ be a geometric random variable with parameter $0 < p \leq 1$. Then $\Pr[X = k] = (1-p)^{k-1}p$ for integer $k \geq 1$. In particular, $\Expect[X] = \frac{1}{p}$.

    \begin{lemma}[Concentration bound for sum of geometric random variables]\label{lem:Chernoff_Geometric}
        Let $X_1, \ldots, X_n$ be independent geometric random variables with parameter $p$, $X = \sum_i X_i$ be their sum and $\mu = \Expect[X] = \frac{n}{p}$. Then for $\delta > 1/p - 1$:
        \[ \Pr[X > \mu + \delta n] \leq e^{-p^2 \delta n/6}  . \]
    \end{lemma}
    
    
    \begin{proof}
        Observe that the event $X > \mu + \delta n$ can be interpreted the event that we need to take more than $\left(\mu+\delta n \right) \Ber(p)$ trials to see $n$ heads. For this to happen, the first $\left(\mu + \delta n \right) \Ber(p)$ trials should see less than $n$ heads, or equivalently, the first $\left(\mu + \delta n \right) \Ber(1-p)$ trials should see greater than $\mu + \delta n - n$ heads. Let $Y = \Bin\left(\mu + \delta n, 1-p\right)$ be the binomial random variable where $\Expect[Y] = (1-p)(\mu + \delta n) = \mu + \delta n - n - p\delta n$. 
        So we have $\Pr[X > \mu + \delta n] \leq \Pr\left[Y > \left(1 + \frac{\delta p n}{\Expect[Y]}\right)\Expect[Y] \right]$. By a Chernoff bound on $Y$,
      \[
            \Pr[X > \mu + \delta n] \leq \Pr\left[Y > \left(1 + \frac{\delta p n}{\Expect[Y]}\right)\Expect[Y] \right]
                                \leq e^{-\frac{\delta^2 p^2 n}{3(1/p + \delta - 1 - p\delta)}} \leq e^{-\delta p^2n/6},
      \]
      where $1/p + \delta - 1 - p\delta = (1/p - 1) + (1 - p)\delta \leq 2\delta$. 
    \end{proof}
}
{}


\begin{definition}[Dependency graph]
    Let $A_1,A_2,\ldots, A_n$ be events in an arbitrary probability space. A directed graph $D = (V, E)$ on the set of vertices $V = \{1,2, \ldots, n\}$ is called a dependency graph for the events $A_1, \ldots, A_n$ if for each $i \in [n]$, the event $A_i$ is mutually independent of all the events $\{A_j : (i,j) \notin E\}$.
\end{definition}

Although $D = (V, E)$ in the above definition is a directed graph, throughout the paper we will only consider the case that $D = (V, E)$  is undirected: $(i,j) \in E$ implies $(i,j) \in E$. In the subsequent discussion, we assume that $D$ is undirected. 

\begin{lemma}[Chernoff bound with bounded dependence~\cite{Pem01}]\label{lem:Chernoff_BD}
    Let $X$ be the sum of $n$ 0-1 random variables $\{X_i\}_{1 \leq i \leq n}$ such that the dependency graph for $\{X_i\}_{1 \leq i \leq n}$ has maximum degree $d$, and let $\mu > \Expect[X]$. Then
    \[ \Pr[X\geq(1+\delta)\mu] \leq O(d)\cdot \exp\left(-\Omega(\delta^2\mu / d)\right). \]
\end{lemma}

\fullornot{
    We make a simple extension of \cref{lem:Chernoff_BD} to the weighted case. Recall that $N^1(v) = N(v) \cup \{v\}$.
    
    \begin{lemma}\label{lem:Chernoff_BD_weighted}
        Let $X_1, X_2, \ldots, X_n$ be $n$ 0-1 random variables and $w_1, w_2,\ldots, w_n$ be positive integers. Consider the weighted sum $X = \sum_i w_i\cdot X_i$ and their dependency graph $G$. Suppose that the numbers  $d$ and $\mu$ satisfy that  $\mu \geq \Expect[X]$ and for each $X_i \in V(G)$, $\sum_{j:X_j \in N^1(X_i)} w_j \leq d$.  Then
        \[ \Pr[X\geq(1+\delta)\mu] \leq O(d)\cdot \exp\left(-\Omega(\delta^2\mu / d)\right). \]
    \end{lemma}
    \begin{proof}
        For each of the random variable $X_i$, we decompose it into $w_i$ 0-1 random variables. As the maximum degree of the new dependency graph is at most $\max_{1 \leq i \leq n} \left\{\sum_{j:X_j \in N^1(X_i)} w_j\right\} \leq d$, the lemma follows from \cref{lem:Chernoff_BD}.
    \end{proof}
}
{}

\fullornot{
    We also generalise the tail bound with bounded dependence to sum of geometric random variables, where the proof idea follows from the proof of \cref{lem:Chernoff_BD} in \cite{Pem01}.
    
    \begin{lemma}\label{lem:geometric_BD}
        Let $X =\sum_{i=1}^n X_i$ be the sum of $n$ geometric random variables with parameter $p$ such that the dependency graph for $X_1, X_2, \ldots, X_n$  has maximum degree $d$, and let $\mu \geq \Expect[X]$, $\delta > 1/p-1$. Then
        \[ \Pr[X\geq \mu + \delta n] \leq O(d)\cdot \exp\left(- p^2 \delta n/ 12d \right)  . \]
    \end{lemma}
    \begin{proof}
        Consider the dependency graph $G$. From $\cite{Pem01}$, $G$ having maxmimum degree $d$ admits an equitable $(d+1)$-vertex-coloring. Such a coloring partitions the vertices into $d+1$ subsets $V_1, V_2,\ldots V_d$, each of size greater than $n/2d$ where the vertices from the same subset form an independent set. Let $\mu_i = \Expect[\sum_{v \in V_i} X_v]$. 
        Note that $X\geq \mu + \delta n$ implies that there exists some $i$ such that $\sum_{v \in V_i} X_v \geq \mu_i + \delta |V_i|$ which happens with probability at most
        \[\Pr\left[\sum_{v \in V_i} X_v \geq \mu_i + \delta |V_i|\right] \leq e^{-p^2 \delta n/12d}  ,\]
        due to \cref{lem:Chernoff_Geometric}. The rest follows by taking a union bound over all $(d+1)$ subsets.
    \end{proof}
    
    Again, we extend the concentration bound to the weighted case. 
    \begin{lemma}\label{lem:geometric_weighted}
        Let $X_1, X_2, \ldots, X_n$ be $n$ geometric random variables with parameter $p$ and $w_1, w_2,\ldots, w_n$ be positive integers.  Let $G$ be the dependency graph for $X_1, X_2, \ldots, X_n$. Consider the weighted sum $X = \sum_{i=1}^n w_i\cdot X_i$.  Suppose the numbers $\mu$ and $d$ satisfy $\mu \geq \Expect[X]$, where $W := \sum_i w_i, \delta > 1/p - 1$, and  $\sum_{j:X_j \in N^1(X_i)} w_j \leq d$ for each $X_i \in V(G)$. Then
        \[ \Pr[X\geq \mu + \delta W] \leq O(d)\cdot \exp\left(-p^2 \delta W/12d\right)  . \]
    \end{lemma}
    \begin{proof}
        For random variable $X_i$, we decompose it into $w_i$ geometric random variables by creating $w_i$ identical copies. Note that the maximum degree of the new dependency graph is at most $\max_{1 \leq i \leq n} \{\sum_{j:X_j \in N^1(X_i)} w_j\} \leq d$. The result follows by applying \cref{lem:geometric_BD}.
    \end{proof}
}
{}
\section{Lower Bounds}\label{sect:lowerbounds}

In this section, we prove that $(1\pm \eps)$-approximate solutions for maximum independent set, maximum cut, minimum vertex cover, and minimum dominating set requires $\Omega\left(\frac{\log n}{\eps}\right)$ rounds to compute in the $\LOCAL$ model, showing that our algorithms are nearly optimal for these problems. We only focus on the scenario when $\eps$ is small in that $0 < \eps \leq \eps_0$ for some universal constant $\eps_0 > 0$. We do not attempt to optimize the choice of $\eps_0$.

Our lower bounds hold even for randomized algorithms whose approximation guarantee only holds in expectation. By standard reductions, our lower bounds also apply to randomized algorithms that succeed with high probability and to deterministic algorithms.

In~\cite{BHKK16}, it was shown that for any $t \leq C \cdot \log n$, where $C > 0$ is some universal constant, any $t$-round randomized algorithm for approximate maximum independent set must have an approximation factor of $n^{-\Omega(t^{-1})}$, in expectation. Setting $t = \Theta(\log n)$, we obtain that a constant-approximation of maximum independent set requires $\Omega(\log n)$ rounds to compute.

Our lower bounds are achieved by extending the lower bound in~\cite{BHKK16} via a series of reductions. Before presenting our proofs, we first review the lower bound proof of~\cite{BHKK16}, which uses the Ramanujan graphs constructed in~\cite{lubotzky1988ramanujan}.

\begin{theorem}[\cite{lubotzky1988ramanujan}]\label{thm:ramanujan}
For any two unequal primes $p$ and $q$ congruent to $1 \mod 4$, there exists a $(p+1)$-regular graph $X^{p,q}$ satisfying the following properties.
\begin{description}
\item[Case 1:] $\paren{\dfrac{q}{p}} = -1$.
\begin{itemize}
    \item $X^{p,q}$ is a bipartite graph with $n = q(q^2 - 1)$ vertices.
    \item The girth of $X^{p,q}$ is at least $4 \log_p q - \log_p 4$.
\end{itemize}
\item[Case 2:] $\paren{\dfrac{q}{p}} = 1$.
\begin{itemize}
    \item $X^{p,q}$ is a non-bipartite graph with $n = q(q^2 - 1)/2$ vertices.
    \item The girth of $X^{p,q}$ is at least $2 \log_p q$.
    \item The size of a maximum independent set of $X^{p,q}$ is at most $\frac{2 \sqrt{p}}{p+1} \cdot n$.
\end{itemize}
\end{description}
\end{theorem}

In \cref{thm:ramanujan}, $\paren{\dfrac{q}{p}} \bydef q^{\frac{p-1}{2}} \ (\mod p) \in \{-1,0,1\}$ is the Legendre symbol.
Note that for any fixed prime $p$ congruent to $1 \mod 4$, the families of graphs $X^{p,q}$ in the above case 1 and case 2 are infinite. For example, for the case $p = 17$, any prime number $q$ congruent to $5 \mod 68$ satisfy (i) $\paren{\dfrac{q}{p}} = -1$ and (ii) $q = 1 \ (\mod 4)$, so $X^{p,q}$ satisfies the properties in the  above case 1. By Dirichlet's theorem on arithmetic progressions, there are infinitely many  prime numbers $q$ such that $q = 5 \ (\mod 68)$. Moreover, by an extension of Bertrand's Postulate to arithmetic progressions~\cite{moree93}, there is a constant $C > 1$ such that for any positive integer $x$, there is a prime $q = 5 \ (\mod 68)$ in the interval $[x, Cx]$.
Similarly, any prime number $q$ congruent to $1 \mod 68$ satisfy (i) $\paren{\dfrac{q}{p}} = 1$ and (ii) $q = 1 \ (\mod 4)$, so $X^{p,q}$ satisfies the properties in the  above case 1, and there are infinitely many such primes $q$.

In the subsequent discussion, we fix $p = 17$, so $X^{p,q}$ is a $18$-regular graph. If $\paren{\dfrac{q}{17}} = -1$ (case 1), then $X^{17,q}$ is bipartite, so the size of a maximum independent set of $X^{17,q}$ is $\frac{n}{2}$. If $\paren{\dfrac{q}{17}} = 1$, then the size of a maximum independent set of $X^{17,q}$ is at most $\frac{2 \sqrt{17}}{17+1} \cdot n < 0.92 \cdot \frac{n}{2}$.

Since the girth of $X^{17,q}$ is $\Omega(\log n)$, any $o(\log n)$-round algorithm is not able to distinguish between case 1 and case 2, as the $o(\log n)$-radius neighborhood of each vertex in $X^{17,q}$ is a $18$-regular complete tree, regardless of $q$. Intuitively, this means that any $o(\log n)$-round algorithm for approximate maximum independent set must have an approximation ratio of at least $0.92$. We formalize the argument in the following theorem.

\begin{theorem}[\cite{BHKK16}]\label{thm:independent_set_basic}
Let $\Algo$ be any randomized algorithm that computes an independent set $I \subseteq V$  such that  $\Expect[|I|] \geq 0.92 \cdot \frac{|V|}{2}$ for any $18$-regular bipartite graph $G=(V,E)$. Then the round complexity of $\Algo$ is $\Omega(\log n)$.
\end{theorem}
\begin{proof}
Suppose the round complexity of $\Algo$ is $o(\log n)$. Then there exists a prime number $q$ congruent to $1 \mod 4$ such that $\paren{\dfrac{q}{17}} = -1$ and $\Algo$ finds an independent set $I \subseteq V(X^{17,q})$ of $X^{17,q}$  such that  $\Expect[|I|] \geq 0.92 \cdot \frac{|V(X^{17,q})|}{2}$ in less than $\frac{\girth(X^{17,q})}{2} - 1$ rounds of communication.

Let $t < \frac{\girth(X^{17,q})}{2} - 1$ be the round complexity of  $\Algo$ in $X^{17,q}$. As the $t$-radius neighborhood of each vertex $v$ in  $X^{17,q}$ is isomorphic to the depth-$t$ complete $18$-regular tree, the probability for $v$ to join the independent set $I$ is identical for all $v$ in $X^{17,q}$. We write $p^\ast$ to denote this probability. By the linearity of expectation,  $\Expect[|I|] = |V(X^{17,q})| \cdot p^\ast$, so we must have $p^\ast \geq \frac{0.92}{2}$.

We  pick another prime number  $q'$ congruent to $1 \mod 4$ such that $\paren{\dfrac{q}{17}} = 1$. We choose $q'$ to be sufficiently large so that $\girth(X^{17,q'}) \geq \girth(X^{17,q})$. Now we run the same $t$-round algorithm in $X^{17,q'}$. As $\girth(X^{17,q'}) \geq \girth(X^{17,q})$, the $t$-radius neighborhood of each vertex in $X^{17,q'}$ is isomorphic to the $t$-radius neighborhood of each vertex in $X^{17,q}$, so the  probability for each vertex $v$ in $X^{17,q'}$ to join the independent set $I$ is also $p^\ast$. Therefore, the output of the algorithm also satisfies $\Expect[|I|] \geq 0.92 \cdot \frac{|V(X^{17,q'})|}{2}$, contradicting the fact that  the size of the maximum independent set of $X^{17,q'}$ is smaller than $0.92 \cdot \frac{|V(X^{17,q'})|}{2}$. Hence the round complexity of $\Algo$ is $\Omega(\log n)$.
\end{proof}

Since the size of a maximum independent set of a regular bipartite graph $G=(V,E)$ is $\frac{|V|}{2}$, \cref{thm:independent_set_basic} shows that any algorithm that finds an $0.92$-approximate maximum independent set in expectation needs $\Omega(\log n)$ rounds. By making $p$ a variable in the above proof, an
$\Omega(\log_p n)$-round lower bound for $O({p^{-0.5}})$-approximation is obtained~\cite{BHKK16}.
In the following theorem, we generalize \cref{thm:independent_set_basic} to $(1-\eps)$-approximation for small $\eps$.

\begin{theorem}\label{thm:independent_set_lb}
Let $0 < \eps \leq 0.04$.
Let $\Algo$ be any randomized algorithm that computes an independent set $I \subseteq V$  such that  $\Expect[|I|] \geq \alpha(G) - \eps \cdot |V|$ for any  graph $G=(V,E)$, where $\alpha(G)$ is the size of a maximum independent set of $G$. Then the round complexity of $\Algo$ is $\Omega\left(\frac{\log n}{\eps}\right)$.
\end{theorem}
\begin{proof}
Suppose the round complexity of $\Algo$ is $o\left(\frac{\log n}{\eps}\right)$. We will use $\Algo$ as a black box to design an $o(\log n)$-round algorithm $\Algo'$ that computes an independent set $I \subseteq V$  such that  $\Expect[|I|] \geq \alpha(G) - 0.04 \cdot |V| \geq 0.92 \cdot \frac{|V|}{2}$ for any $18$-regular bipartite graph $G=(V,E)$, contradicting \cref{thm:independent_set_basic}, so  the round complexity of $\Algo$ must be $\Omega\left(\frac{\log n}{\eps}\right)$.

Let $G=(V,E)$ be any $18$-regular bipartite graph, and let $x$ be a non-negative integer.
We define $G^x$ as the result of subdividing each edge $e=\{u,v\} \in E$ into a path $P_e$ of length $2x+1$: $(u, w_1, w_2, \ldots, w_{2x}, v)$. It is clear that $G^x$ is a bipartite graph with $(18x+1) \cdot |V|$ vertices, so the size of a maximum independent set of $G^x$ is  $\frac{18x+1}{2} \cdot |V|$.

We set $x \bydef \floor{\frac{0.08 \cdot \eps^{-1} - 1}{18}}$. If $\eps > \frac{2}{475}$, then $x = 0$, in which case we have $G^x = G$. Otherwise, $x = \Theta\left(1/\eps\right)$. In any case, our choice of $x$ satisfies that $\eps \cdot (18x+1) \leq 0.08$.

We are ready to describe the algorithm $\Algo'$. Given any  $18$-regular bipartite graph $G=(V,E)$, we simulate the virtual graph $G^x$ and run $\Algo$ in $G^x$. By our choice of $x$, the simulation costs $o(\log n)$ rounds in $G$. Let $I^{\diamond}$ be the independent set of $G^x$ returned by $\Algo$. We compute an independent set $I$ of $G$ as follows. Each vertex $v \in V$ generates a distinct identifier $\ID(v)$ using its local random bits. For each $v \in V$, we add $v$  to $I$ if the following two conditions are met.
\begin{itemize}
    \item $v \in I^{\diamond}$.
    \item  For each neighbor $u$ of $v$, either $u \notin I^{\diamond}$ or $\ID(v) < \ID(u)$.
\end{itemize}
It is clear that $I$ is an independent set of $G$. We lower bound the expected size of $I$ as follows.
Let $E'$ be the set of edges $e \in E$ such that both endpoints of $e$ are in $I^{\diamond}$. We make the following observations. 
\begin{itemize}
    \item For each $e = \{u,v\} \in E'$, the number of degree-2 vertices of $P_e$ in $I^{\diamond}$ is at most $x-1$.
    \item For each $e = \{u,v\} \in E \setminus E'$, the number of degree-2 vertices of $P_e$ in $I^{\diamond}$ is at most $x$.
\end{itemize}

We write $k_1$ to denote the number of degree-2 vertices of $G^x$ in $I^{\diamond}$, and we write $k_2$ to denote the number of degree-18 vertices in $I^{\diamond} \setminus I$, so $|I| = |I^{\diamond}| - k_1 - k_2$. We have $k_1 \leq x \cdot (|E| - |E'|) + (x-1) \cdot |E'|$ by the above observations and $k_2 \leq |E'|$, so $|I| \geq |I^{\diamond}| - x \cdot |E| = |I^{\diamond}| - 9x \cdot |V|$. Now we are ready to lower bound $\Expect[|I|]$.
\begin{align*}
    \Expect[|I|]
    &\geq \Expect[|I^{\diamond}|] - 9x \cdot |V|\\
    &\geq (1-\eps) \cdot \frac{18x+1}{2} \cdot |V|  - 9x \cdot |V|\\
    &=\frac{|V|}{2} - \eps \cdot \frac{18x+1}{2} \cdot |V|\\
    &\geq \frac{|V|}{2} - 0.08 \cdot \frac{|V|}{2}\\
    &= 0.92 \cdot \frac{|V|}{2}. \qedhere
\end{align*}
\end{proof}

Since $\alpha(G) - \eps \cdot |V| \leq (1-\eps)  \cdot\alpha(G)$, \cref{thm:independent_set_lb} gives an  $\Omega\left(\frac{\log n}{\eps}\right)$ lower bound for computing an $(1-\eps)$-approximate maximum independent set.
Since a vertex cover is, by definition, a complement of an independent set, the same lower bound applies to $(1+\eps)$-approximate minimum vertex cover.

\begin{theorem}\label{thm:vertex_cover_lb}
Let $0 < \eps \leq 0.04$.
Let $\Algo$ be any randomized algorithm that computes a vertex cover $S \subseteq V$  such that  $\Expect[|S|] \leq \tau(G) + \eps \cdot |V|$ for any  graph $G=(V,E)$, where $\tau(G)$ is the size of a  minimum vertex cover of $G$. Then the round complexity of $\Algo$ is $\Omega\left(\frac{\log n}{\eps}\right)$.
\end{theorem}
\begin{proof}
Observe that $I \bydef V \setminus S$ is an independent set with $\Expect[|I|] = |V| - \Expect[|S|] \geq |V| - \tau(G) - \eps \cdot |V| = \alpha(G) - \eps \cdot |V|$, so the round complexity of $\Algo$ is $\Omega\left(\frac{\log n}{\eps}\right)$ by  \cref{thm:independent_set_lb}.
\end{proof}

By a standard reduction from minimum dominating set to minimum vertex cover, we also obtain an  $\Omega\left(\frac{\log n}{\eps}\right)$ lower bound for  $(1+\eps)$-approximate minimum dominating set.

\begin{theorem}\label{thm:dominating_set_lb}
Let $0 < \eps \leq 0.04$.
Let $\Algo$ be any randomized algorithm that computes a dominating set $D \subseteq V$  such that  $\Expect[|D|] \leq (1+\eps)\cdot \gamma(G)$ for any  graph $G=(V,E)$, where $\tau(G)$ is the size of a  minimum dominating set of $G$. Then the round complexity of $\Algo$ is $\Omega\left(\frac{\log n}{\eps}\right)$.
\end{theorem}
\begin{proof}
Let $G=(V,E)$ be any graph. We define $G^\ast = (V^\ast, E^\ast)$ as the result of the following modification to $G$: For each edge $e=\{u,v\} \in E$, add a new vertex $w_e$ and two new edges $\{w_e, u\}$ and $\{w_e, v\}$.

We claim that $\gamma(G^\ast) = \tau(G)$. The fact that  $\gamma(G^\ast) \leq \tau(G)$ follows from the observation that any vertex cover of $G$ is also a dominating set of $G^\ast$. For the other direction, for any given  dominating set $D$ of $G^\ast$, we can obtain a vertex cover $S$ of $G$ such that $|S| \leq |D|$, as follows. Initialize $S = D$. For each edge $e=\{u,v\}  \in E$ such that $w_e \in D$, remove $w_e$ from $D$ and add $u$ to $D$. Since the set of neighbors of $w_e$ is a subset of the set of neighbors of $u$, the resulting set $S \subset V$ is still a dominating set of $G^\ast$. For each edge $e=\{u,v\}  \in E$, since $w_e$ is dominated by $S$, at least one of $\{u,v\}$ is in $S$, so $S$ is a vertex cover of $G$. Hence $\gamma(G^\ast) \geq \tau(G)$.

Suppose the round complexity of $\Algo$ is $o\left(\frac{\log n}{\eps}\right)$. Then we may use  $\Algo$  as a black box to design an  $o\left(\frac{\log n}{\eps}\right)$-round algorithm $\Algo'$ that computes a vertex cover $S \subseteq V$  such that  $\Expect[|S|] \leq (1+\eps)\tau(G)  \leq \tau(G) + \eps \cdot |V|$ for any  graph $G=(V,E)$, contradicting \cref{thm:vertex_cover_lb}.

The algorithm $\Algo'$ simply simulates $\Algo$ on $G^\ast$ and applies the above transformation to turn the dominating set $D$ of $G^\ast$ returned by  $\Algo'$ into a vertex cover $S$ of $G$. We have $\Expect[|S|] \leq \Expect[|D|] \leq (1+\eps)\cdot \gamma(G^\ast) =  (1+\eps)\cdot \tau(G)$.
\end{proof}

Next, we consider the maximum cut problem, whose goal is to find a bipartition of the vertex set $V = S \cup T$ such that the number of edges $E^\ast = \{ e=\{u,v\} \ : \ u\in S, v \in T\}$ crossing the two parts is maximized. There are two natural versions of the problem in the distributed setting.
\begin{enumerate}
    \item By the end of the computation, each vertex $v \in V$ decides whether $v \in S$ or $v \in T$.
    \item By the end of the computation, each edge $e \in E$ decides whether $e \in E^\ast$ or $e \notin E^\ast$.
\end{enumerate}

Given the bipartition $V = S \cup T$, each edge $e \in E$ can locally decide whether $e \in E^\ast$ or $e \notin E^\ast$, so the first version is at least as hard as the second version, Therefore, we will focus on proving lower bounds for the second version.

\begin{lemma}\label{lem:max_cut_aux}
Let $q$ be any prime congruent to $1 \mod 4$ such that  $\paren{\dfrac{q}{p}} = 1$. Then the size of maximum cut of $X^{17,q}$ is smaller than $0.999 \cdot |E(X^{17,q})|$.
\end{lemma}
\begin{proof}
By \cref{thm:ramanujan}, we know that the size of a maximum independent set of $X^{17,q}$ is at most  $\frac{2 \sqrt{17}}{17+1} \cdot |V(X^{17,q})| < 0.92 \cdot \frac{|V(X^{17,q})|}{2}$. Therefore, to prove the lemma, it suffices to prove the following claim.
\begin{itemize}
    \item For any graph $G=(V,E)$, if it has a cut with at least $|E| - x$ edges, then it admits an independent set of size at least $\frac{|V| - x}{2}$.
\end{itemize}
To prove this lemma by the above claim, suppose the size of maximum cut of $X^{17,q}$ is at least $0.999 \cdot |E(X^{17,q})| = |E(X^{17,q})| - x$ with $x = 0.001 \cdot |E(X^{17,q})| = 0.009 \cdot |V(X^{17,q})|$, as $X^{17,q}$ is $18$-regular. Then the above claim implies that  $X^{17,q}$ contains an independent set of size at least $\frac{|V(X^{17,q})| - x}{2} = 0.991 \cdot \frac{|V(X^{17,q})|}{2} > 0.92 \cdot \frac{|V(X^{17,q})|}{2}$, contradicting \cref{thm:ramanujan}.

We now prove the above claim. Consider a  cut $E^\ast$ of  $G=(V,E)$ that contains at least $|E| - x$ edges. We find an independent set $I$ of size at least $\frac{|V| - x}{2}$, as follows. For each edge $e \in E \setminus E^\ast$, we remove one of its two endpoints from the graph. Since $|E^\ast| \geq |E| - x$, at most $x$ vertices are removed. After removing these vertices, the graph becomes bipartite, i.e., the remaining vertices can be partitioned into two parts $X$ and $Y$ such that both $X$ and $Y$ are independent sets, and we have $\max\{|X|, |Y|\} \geq \frac{|V| - x}{2}$. If $|X| \geq |Y|$, we may set $I = X$. Otherwise, we may set $I = Y$.
\end{proof}

Using \cref{lem:max_cut_aux}, the $\Omega\left(\frac{\log n}{\eps}\right)$ lower bound for $(1-\eps)$-approximate maximum cut can be proved via the approach of \cref{thm:independent_set_basic,thm:independent_set_lb}.

\begin{theorem}\label{thm:max_cut_basic}
Let $\Algo$ be any randomized algorithm that computes a cut $E^\ast \subseteq E$  such that  $\Expect[|E^\ast|] \geq 0.999 \cdot |E|$ for any $18$-regular bipartite graph $G=(V,E)$. Then the round complexity of $\Algo$ is $\Omega(\log n)$.
\end{theorem}
\begin{proof}
Suppose the round complexity of $\Algo$ is $o(\log n)$. Then there exists a prime number $q$ congruent to $1 \mod 4$ such that $\paren{\dfrac{q}{17}} = -1$ and $\Algo$ finds a cut $E^\ast \subseteq E(X^{17,q})$ of $X^{17,q}$  such that  $\Expect[|E^\ast|] \geq 0.999 \cdot |E(X^{17,q})|$ in less than $\frac{\girth(X^{17,q})}{2} - 1$ rounds of communication.

Let $t < \frac{\girth(X^{17,q})}{2} - 1$ be the round complexity of  $\Algo$ in $X^{17,q}$. As the $t$-radius neighborhood of each vertex $v$ in  $X^{17,q}$ is isomorphic to the depth-$t$ complete $18$-regular tree, the probability for an  edge $e$ to join the cut $E^\ast$ is identical for all edges $e$ in $X^{17,q}$. We write $p^\ast$ to denote this probability. By the linearity of expectation,  $\Expect[|E^\ast|] = |E(X^{17,q})| \cdot p^\ast$, so we must have $p^\ast \geq 0.999$.

We  pick another prime number  $q'$ congruent to $1 \mod 4$ such that $\paren{\dfrac{q}{17}} = 1$. We choose $q'$ to be sufficiently large such that $\girth(X^{17,q'}) \geq \girth(X^{17,q})$. Now we run the same $t$-round algorithm in $X^{17,q'}$. As $\girth(X^{17,q'}) \geq \girth(X^{17,q})$, the $t$-radius neighborhood of each vertex in $X^{17,q'}$ is isomorphic to the $t$-radius neighborhood of each vertex in $X^{17,q}$, so the  probability for each edge $e$ in $X^{17,q'}$ to join the cut $E^\ast$ is also $p^\ast$. Therefore, the output of the algorithm also satisfies $\Expect[|E^\ast|] \geq 0.999 \cdot |E(X^{17,q'})|$, contradicting \cref{lem:max_cut_aux}. Hence the round complexity of $\Algo$ is $\Omega(\log n)$.
\end{proof}

Since the size of a maximum cut of a  bipartite graph $G=(V,E)$ is $|E|$, \cref{thm:max_cut_basic} shows that any algorithm that finds a $0.999$-approximate maximum cut in expectation needs $\Omega(\log n)$ rounds. Similar to \cref{thm:independent_set_basic}, we may generalize the lower bound to $(1-\eps)$-approximation for small $\eps$.

\begin{theorem}\label{thm:max_cut_lb}
Let $0 < \eps \leq 0.001$.
Let $\Algo$ be any randomized algorithm that computes a  cut $E^\ast \subseteq E$  such that  $\Expect[|E^\ast|] \geq (1-\eps) \cdot |E|$ for any $18$-regular bipartite graph $G=(V,E)$. Then the round complexity of $\Algo$ is $\Omega\left(\frac{\log n}{\eps}\right)$.
\end{theorem}
\begin{proof}
Suppose the round complexity of $\Algo$ is $o\left(\frac{\log n}{\eps}\right)$. We will use $\Algo$ as a black box to design an $o(\log n)$-round algorithm $\Algo'$ that computes a cut $\tilde{E} \subseteq E$  such that  $\Expect[|\tilde{E}|] \geq 0.999 \cdot |E|$ for any $18$-regular bipartite graph $G=(V,E)$, contradicting \cref{thm:max_cut_basic}, so  the round complexity of $\Algo$ must be $\Omega\left(\frac{\log n}{\eps}\right)$.

Let $G=(V,E)$ be a $18$-regular bipartite graph, and let $x$ be a non-negative integer.
Similar to the proof of \cref{thm:independent_set_basic}, we define $G^x$ as the result of subdividing each edge $e=\{u,v\} \in E$ into a path $P_e$ of length $2x+1$: $(u, w_1, w_2, \ldots, w_{2x}, v)$. It is clear that $G^x$ is a bipartite graph with $(2x+1) \cdot |E|$ edges, so the size of the maximum cut of $G^x$ is  $(2x+1) \cdot |E|$.

We set $x \bydef \floor{\frac{0.001 \cdot \eps^{-1} - 1}{2}}$. Note that we may have $x = 0$ if $\eps < \frac{1}{3000}$, in which case we have $G^x = G$. Otherwise, $x = \Theta\left(1/\eps\right)$. In any case, our choice of $x$ satisfies that $\eps \cdot (2x+1) \leq 0.001$.

We are ready to describe the algorithm $\Algo'$. Given any  $18$-regular bipartite graph $G=(V,E)$, we simulate the virtual graph $G^x$ and run $\Algo$ in $G^x$. By our choice of $x$, the simulation costs $o(\log n)$ rounds in $G$. Let $E^\ast$ be the cut of $G^x$ returned by $\Algo$. We compute a cut $\tilde{E}$ of $G$ as follows.

For each edge $e\in E$, let $K_e$ be the number of edges in $P_e = (u, w_1, w_2, \ldots, w_{2x}, v)$ that are in $E^\ast$. Observe that $K_e$ is an even number if and only if $u$ and $v$ belong to the same side of the cut $E^\ast$ of $G^x$. For each edge $e\in E$, we add $e$ to  $\tilde{E}$ if $K_e$ is an odd number. It is clear that $\tilde{E}$ is a cut that partitions $V$ into two parts.

To lower bound $\Expect[|\tilde{E}|]$, we make the following observations.
\begin{itemize}
    \item If $K_e$ is odd, then $K_e \leq 2x+1$.
    \item If $K_e$ is even, then $K_e \leq 2x$.
\end{itemize}
Since $|\tilde{E}|$ equals the number of $e \in E$ such that $K_e$ is odd, we have \[|E^\ast| \leq (2x+1)\cdot|\tilde{E}| + 2x \cdot (|E|-|\tilde{E}|) = 2x \cdot |E| - |\tilde{E}|,\]
which implies
\begin{align*}
  \Expect[|\tilde{E}|] &\geq  \Expect[|E^\ast|] - 2x \cdot |E| \\
&=  \Expect[|E^\ast|] - (2x+1) \cdot |E| + |E|\\
&\geq - \eps \cdot (2x+1) \cdot |E| + |E|\\
&\geq - 0.001 \cdot |E| + |E|\\
&= 0.999 \cdot |E|.
\end{align*}
In the calculation, we use the fact that $\Algo$ computes an $(1-\eps)$-approximate maximum cut, so $\Expect[|E^\ast|] \geq (1-\eps) \cdot |E(G^x)| = (1-\eps) \cdot (2x+1) \cdot |E|$. The inequality $\eps \cdot (2x+1) \leq 0.001$ is due to our choice of $x$.
\end{proof}

Now we may prove \cref{thm:lowerbound-main}.

\begin{proof}[Proof of \cref{thm:lowerbound-main}]
The theorem follows from \cref{thm:independent_set_lb,thm:vertex_cover_lb,thm:dominating_set_lb,thm:max_cut_lb}.
\end{proof}

\section{Existing Approaches to Low-Diameter Decompositions}\label{sect:LDD}



In this section, we assume an upper bound $\nn \geq n$ on the true number of vertices $n = V(G)$ is initially known to all vertices.
We first review the low-diameter decomposition algorithm of~\cite{EN16}, which deletes $O(\lambda)$ fraction of the vertices in expectation such that each remaining connected component has diameter $\left(\frac{\log \nn}{\lambda}\right)$. For the sake of completeness, we provide a proof sketch of this result. 

\begin{lemma}[\cite{EN16}]\label{thm:vertex_LDD}
    Given a graph $G=(V,E)$ and parameter $\lambda$, there is an algorithm in the $\LOCAL$ model that deletes a fraction of vertices in $\frac{4\ln \nn}{\lambda}$ rounds meeting the following conditions.
    \begin{itemize}
        \item Each remaining connected component has diameter at most  $\frac{8\ln \nn}{\lambda}$.
        \item For each vertex $v \in V$, the probability that $v$ is deleted is at most $1 - e^{-\lambda} + \nn^{-3}$.
    \end{itemize}
\end{lemma}
\begin{proof}
    The algorithm runs as follows: Each vertex $v$ samples a value $T_v$ from the exponential distribution with parameter $\lambda$ and broadcasts $T_v$ to its $\lfloor T_v\rfloor$-hop neighborhood. For any vertex $v$,
    \[ \Pr\left[T_v \geq 4\cdot \frac{\ln \nn}{\lambda}\right] = e^{- \lambda \left(4\cdot \frac{\ln \nn}{\lambda} \right)} \leq \tilde{n}^{-4}.\]
    Should such event happen, the vertex $v$ simply \emph{resets} $T_v = 0$ and proceeds as usual.  
    
    Now we focus on a vertex $v$. With respect to $v$, we order $V=(v_1, v_2, \ldots, v_n)$ in such a way that $m^{(v)}_{v_1} \geq m^{(v)}_{v_2} \geq \cdots \geq m^{(v)}_{v_n}$, where we define 
    \[m^{(v)}_{v_i} := T_{v_i} - \dist(v,v_i).\] 
    There are two cases.
    \begin{itemize}
        \item If $m^{(v)}_{v_2} \geq m^{(v)}_{v_1} - 1$, then $v$ deletes itself.
        \item Otherwise, $v$ joins the cluster of $v_1$.
    \end{itemize}
    Note that if $v$ is outside of the $\lfloor T_u\rfloor$-hop neighborhood of $u$, then $T_u - \dist(v,u) < 0 \leq T_v - \dist(v,v)$, so such a vertex $u$ will never be a candidate for $v$ to join the cluster of $u$. Therefore, each vertex $u$ only needs to broadcast $T_u$ to its $\lfloor T_u\rfloor$-hop neighborhood, and so the running time of this process is upper bounded by $\frac{4\ln \nn}{\lambda}$.
    
\paragraph{Cluster diameter.} Each cluster has \emph{weak diameter} at most $\frac{8\ln \nn}{\lambda}$, as we observe that $v$ joins the cluster of $u$ \emph{only if} $\dist(v,u) \leq T_u \leq \frac{4\ln \nn}{\lambda}$. To see that the same bound holds for \emph{strong diameter},  consider any vertex $v$, an let $S$ be the cluster that $v$ belongs to. Let $P = (u_1, u_2, u_3 \ldots, u_k = v)$ be a shortest path from  $u_k = v$ to the centre $u_1$ of the cluster $S$. We argue that all vertices on the path $P$ belongs to $S$. Assume towards contradiction that some vertex $w$ on the path does not belong to $S$. Then $w$ must receive a value $T_z$ from  a vertex $z$ such that\[m^{(w)}_{z} := T_{z} - \dist(w,z) > T_{u_1} - \dist(w,u_1) - 1 = m^{(w)}_{u_1} - 1,\] for otherwise $w$ would join the subset $S$. 
   Therefore, we have
   \begin{align*}
       m^{(v)}_{z} &:= T_{z} - \dist(z, v) \\
       &\geq T_{z} - \dist(w,z) - \dist(w,v) \\
       &>  T_{u_1} - \dist(w,u_1) - 1 - \dist(w,v) \\
       &\geq T_{u_1} - \dist(v,u_1) - 1\\
       &= m^{(v)}_{u_1} - 1.
   \end{align*}
   contradicting that fact that $v$ joins the cluster $S$ of $u_1$.
    
  \paragraph{Probability of deletion.} For the rest of the proof, we analyze the probability that a vertex $v$ is deleted.
    In the analysis, we consider the version of the above algorithm that  \emph{does not reset} $T_v$ when $T_v$ exceeds $4\cdot \frac{\ln \nn}{\lambda}$. That is, $T_v \sim \exponential(\lambda)$. Since $\Pr\left[T_v \geq 4\cdot \frac{\ln \nn}{\lambda}\right]   \leq \nn^{-4}$, the probability that the version of algorithm without resetting behaves identically to the original algorithm is at least $1 - \nn^{-3}$, Therefore, to prove that the probability that $v$ is deleted is at most $1 - e^{-\lambda} + \nn^{-3}$ in the original algorithm, we just need to show that this  probability  is at most $1 - e^{-\lambda}$ in the version that we do not reset $T_v$.
    
    The probability that $v$ is deleted equals the probability that  $m^{(v)}_{v_2} \geq m^{(v)}_{v_1} - 1$. We analyze the probability by first reveal the value of $m^{(v)}_{v_2}$. After fixing $m^{(v)}_{v_2}$, by the memoryless property of the exponential distribution, we have \[\Pr\left[m^{(v)}_{v_1} - m^{(v)}_{v_2} \leq 1\right] \leq \Pr\left[\exponential(\lambda) \leq 1\right] = 1 - e^{-\lambda},\] so   $v$ is deleted  with  probability    at most $1 - e^{-\lambda}$, as required.
\end{proof}



\fullornot{
    We  modify the above low-diameter decomposition algorithm to one that finds a sparse cover for all hyperedges of any given hypergraph. As we will later see, this decomposition will be useful in solving covering ILP problems. In the following lemma, recall that we say a random variable $X$ is dominated by $Y$ if there is a coupling between $X$ and $Y$ such that $X \leq Y$. In particular, $\Pr[X \geq k] \leq \Pr[Y \geq k]$ for all $k$. Note that the weak diameter bound in the following lemma can be strengthened to a strong diameter bound via a proof similar to the one for \cref{thm:vertex_LDD}, but a weak diameter bound suffices for our application in this paper.
    
    \begin{lemma}\label{thm:LDD_covering}
    Given a hypergraph $H=(V,E)$ and a parameter $\lambda$, there is an algorithm that achieves the following in the $\LOCAL$ model in $\frac{4\ln \nn}{\lambda}$ rounds.
    \begin{itemize}
        \item The algorithm computes subsets $S_1, \ldots, S_k$ of $V$ such that the weak diameter of each $S_i$ is at most $\frac{8\ln \nn}{\lambda}$.
        \item Each hyperedge $e \in E$ is completely contained in at least one induced subgraph $H(S_i)$.
        \item For each vertex $v$, let $X_v$ denote the number of clusters $S_i$ that contains $v$. Then $X_v$ is dominated by $\geom\left(e^{-\lambda}\right) + \nn^{-2}$.
    \end{itemize} 
    \end{lemma}
    
    \begin{proof}
        We consider the same process of generating $T_v$ and broadcasting $T_v$ as in the algorithm of \cref{thm:vertex_LDD}. The only difference here is that
        here $v$ never delete itself, and we let $v$ joins the cluster of $v_j$ for all $v_j$ such that $m^{(v)}_{v_j}\geq m^{(v)}_{v_1} - 1$. For the sake of simplicity, in the subsequent discussion we write $m^{(v)}_i  = m^{(v)}_{v_i}$. Similar to  \cref{thm:vertex_LDD}, the weak diameter of each cluster $S_i$ is at most $\frac{8\ln \nn}{\lambda}$, as we observe that $v$ joins the cluster of $u$ \emph{only if} $\dist(v,u) \leq T_u \leq \frac{4\ln \nn}{\lambda}$. 
        
        \paragraph{Each hyperedge is covered.}
        To show that each hyperedge $e$ is covered by at least one cluster, we consider any hyperedge $e = \{u_1, u_2, \ldots\}$, and let $u_k \in e$ be the vertex that maximises $m^{(u_k)}_1$. Let $v_1, v_2, \ldots, v_n$ be the ranking of $V$ from the perspective of $u_k$. Recall that $m^{(u_k)}_1 = T_{v_1} - \dist(u_k, v_1)$. For any $u_j \in e$, we have $T_{v_1} - \dist(u_j, v_1) \geq m^{(u_k)}_1 - 1 \geq m^{(u_j)}_1 - 1$. Hence $u_j$ must join the cluster of $v_1$, so $e$ is fully contained in the cluster of $v_1$.
    
        
        \paragraph{Probability of deletion.}
      For the rest of the proof, we analyze the probability of $X_v$.
        Similar to the proof of \cref{thm:vertex_LDD}, in the analysis we consider the version of the  algorithm that  \emph{does not reset} $T_v$ when $T_v$ exceeds $4\cdot \frac{\ln \nn}{\lambda}$. As discussed in the proof of  \cref{thm:vertex_LDD}, the probability that the version of algorithm without resetting behaves identically to the original algorithm is at least $1 - \nn^{-3}$, Therefore, to show that $X_v$ is dominated by $\geom\left(e^{-\lambda}\right) + \nn^{-2}$ in the original algorithm, we just need to show that  $X_v$ is dominated by $\geom\left(e^{-\lambda}\right)$  in the version that we do not reset $T_v$. The additive term $\nn^{-2}$ is to account for the bad event which happens with probability at most $\nn^{-3}$, and here we may use the trivial upper bound $|V| \leq \nn$ on $X_v$.
        
        Note that $X_v$ is exactly the  number of indices $j$ such that $m^{(v)}_j \geq m^{(v)}_1 - 1$. 
        For any integer $t \geq 1$, $X_v \geq t$ implies $m^{(v)}_t+1 \geq m^{(v)}_1$. Now, for any 
        value $a$, conditioning on the event $m^{(v)}_t = a$,
        \begin{align*}
        \Pr\left[m^{(v)}_1 \leq m^{(v)}_t+1 \; \middle| \; m^{(v)}_t = a\right] &= \prod_{i = 1}^{t-1} \Pr\left[m^{(v)}_i \leq a+1 \; \middle| \; m^{(v)}_i \geq a\right] \\
        &\leq \left(1 - e^{-\lambda}\right)^{t-1}  ,
        \end{align*}
        due to the memoryless property of the exponential distribution, and this calculation is independent of $a$. 
        Therefore,
        \[ \Pr\left[X_v\geq t\right] = \Pr\left[m^{(v)}_1 \leq m^{(v)}_t + 1\right] \leq \left(1 - e^{-\lambda}\right)^{t-1}. \qedhere\]
    \end{proof}
}
{}



\fullornot{
    Making use of the above result,  we may solve any covering ILP problem in the following manner. 
    
    \begin{lemma}\label{thm:simple_covering}
        Let $H=(V,E)$ be a hypergraph representing a covering ILP problem and let $0 < \lambda < 1$ be a parameter.
        Suppose we are given a sparse cover $S_1, S_2, \ldots, S_k$ with the following properties.
    \begin{itemize}
        \item Each $S_i$ is a subset of $V$ with weak diameter at most $O\left(\frac{\log \nn}{\lambda}\right)$.
        \item Each hyperedge $e \in E$ is completely contained in at least one induced subgraph $H(S_i)$.
    \end{itemize}
        Then there is an algorithm that takes $O\left(\frac{\log \nn}{\lambda}\right)$ rounds in the $\LOCAL$ model and finds a solution of the covering ILP problem whose weight is at most $\sum_{v\in V} X_v \cdot Q^*(v) \cdot w_v$, where $Q^*$ is any fixed optimal solution and $X_v$ is the number of clusters $S_i$ that contains $v$. 
    \end{lemma}
    \begin{proof}
        For each subset $S_i$, compute the local optimal covering solution $Q^{\local}_{S_i}$ on $S_i$. We combine the local solutions into a global solution in the following manner. If a variable is assigned one in any local solution, it will be assigned one in the global solution. Otherwise, it is assigned zero. In other words, we are taking a bit-wise OR on the solution vectors. 
        
        As each hyperedge $e$ is contained in at least one $H(S_i)$, the constraint corresponding to $e$ is satisfied by $Q^{\local}_{S_i}$ in the local instance defined by $S_i$. The way we obtain a global solution from combining local ones implies that all constraints are satisfied by the global solution. It remains to bound the size of the such solution, which is at most
        \begin{align*}
        \sum_i W(Q^{\local}_{S_i}, S_i) &\leq \sum_i W(Q^*,S_i) \\
                                        &\leq \sum_i \sum_{v\in S_i} Q^*(v) \cdot w_v \\
                                     &= \sum_{v\in V} X_v \cdot Q^*(v) \cdot w_v. \qedhere
        \end{align*}
    \end{proof}
    
}
{}

\fullornot{
\subsection{Limitations of Existing Approaches}

    
    Miller, Peng, and Xu~\cite{MPX13} designed  an $O\left(\frac{\log n}{\eps}\right)$-round algorithm that deletes at most $\eps |E(G)|$ edges \emph{in expectation} such that each remaining connected component has diameter  $O\left(\frac{\log n}{\eps}\right)$. Elkin and Neiman extends the approach of~\cite{MPX13} to design an  $O\left(\frac{\log n}{\eps}\right)$-round algorithm that deletes $\eps |V(G)|$ vertices \emph{in expectation} such that each remaining connected component has diameter  $O\left(\frac{\log n}{\eps}\right)$.

    In this section, we present families of graphs such that the number of deleted vertices (resp., edges) exceed $\eps |V(G)|$ (resp., $\eps |E(G)|$) with non-negligible probability, if we run the low-diameter decomposition algorithms~\cite{EN16,MPX13} on them.
    
    
    For completeness, we start with a description of the algorithms from \cite{EN16,MPX13}. In particular, we highlight their deletion behaviour: Each vertex $v$ samples a value $T_v\sim \exponential(\eps)$ from the exponential distribution and broadcasts $T_v$ to its $\lfloor T_v\rfloor$-hop neighborhood. For any vertex $v$, we order $V=(v_1, v_2, \ldots, v_n)$ in such a way that $m^{(v)}_{v_1} \geq m^{(v)}_{v_2} \geq \cdots \geq m^{(v)}_{v_n}$, where we define 
        \[m^{(v)}_{v_i} := T_{v_i} - \dist(v,v_i).\] 
        Here is how the deletion happens in the two algorithms:
        \begin{itemize}
            \item In the Elkin-Neiman algorithm~\cite{EN16}, each vertex $v$ deletes itself if $m^{(v)}_{v_2} \geq m^{(v)}_{v_1} - 1$.
            \item In the Miller-Peng-Xu algorithm~\cite{MPX13}, each vertex $v$ joins the cluster of $v^\ast$ such that $m_{v^\ast}^{(v)}$ is maximized, and each edge $e=\{u,v\}$ deletes itself if $u$ and $v$ joins different clusters.
        \end{itemize}
    
    \begin{claim}\label{claim:EN}
        There exists a family of graphs such that when we run the Elkin-Neiman algorithm, at least $n-1$ vertices are deleted with probability $\Omega(\eps)$.
    \end{claim}
    \begin{proof}
        Consider the clique of $n$ vertices. Recall that each vertex $v$ samples $T_v\sim \exponential(\eps)$ independently. We order $V = (w_1, w_2, \ldots, w_n)$ in such a way that $T_{w_1} \geq T_{w_2} \geq \cdots \geq T_{w_n}$. We claim that if \[T_{w_1} \leq T_{w_2} + 1,\] then all vertices $w_i$ with $i\geq 2$ are  deleted. Note that $w_i$ receives values $T_{w_1} - 1 \geq T_{w_2} - 1 \geq \cdots \geq T_{w_n} - 1$ (excluding $T_{w_i} - 1$) as well as $T_{w_i}$, as $\dist(w_i, w_j) = 1$ for $i \neq j$. There are two cases.
        \begin{itemize}
            \item $T_{w_i} \geq T_{w_1}-1$: We know $m^{(w_i)}_{w_1} = T_{w_1} - 1 \geq T_{w_i} - 1 = m^{(w_i)}_{w_i} - 1$,  so $w_i$ is deleted.
            \item $T_{w_i} < T_{w_1}-1$: We know $m^{(w_i)}_{w_2} = T_{w_2} - 1 \geq   T_{w_1} - 2 > T_{w_i} - 1 \geq m^{(w_i)}_{w_i} - 1$,  so  $w_i$ is deleted. 
        \end{itemize}
        The probability of the event $T_{w_1} \leq T_{w_2} + 1$ happening is 
        \[\Pr\left[T_{w_1} \leq T_{w_2} + 1 \; \middle| \; T_{w_1}\geq T_{w_2}\right] = 1 - e^{-\eps} = \Omega(\eps),  \]
        by the memoryless property of the exponential distribution.
    \end{proof}
    
    \begin{claim}\label{claim:MPX}
        There exists a family of graphs such that when we run the  Miller-Peng-Xu algorithm, at least $(1 - O(1/n))|E(G)|$ edges are deleted with probability   $\Omega(\eps)$.
    \end{claim}
    \begin{proof}
        Consider the following construction of a graph $G=(V,E)$ with $n = 4t+2$ vertices and $m = t^2 + 4t$ edges. Let $S_L, S_R, L, R$ be vertex sets of size $t$, and let $u$ and $v$ be two vertices. Define $V := \{u\} \cup \{v\} \cup S_L \cup S_R \cup L \cup R$. Add an edge connecting each vertex in $L$ and each vertex in $R$ to make $(L,R)$ a complete bipartite graph. Add an edge between $u$ and each vertex in $S_L\cup L$. Add an edge between $v$ and each vertex in $S_R \cup R$. 
        
        Recall that each vertex $v$ samples $T_v\sim \exponential(\eps)$ independently. We order $V = (w_1, w_2, \ldots, w_n)$ in such a way that $T_{w_1} \geq T_{w_2} \geq \cdots \geq T_{w_n}$. Let us consider the following event $\Event$: 
    \[  w_1 \in S_L, \ w_2 \in S_R, \  T_{w_2}>T_{w_3}+2, \ \text{and} \ T_{w_1} < T_{w_2} + 1.\]    
     If  event $\Event$ happens, then  we know that all $t^2$ edges between $L$ and $R$ are going to be deleted by the algorithm, because all vertices in  $\{u\} \cup  L \cup S_L$ will join the cluster of $w_1$ and all vertices in  $\{v\} \cup R \cup S_R$ will join the cluster of $w_2$.
     The fraction of removed edges is $\frac{t^2}{t^2 + 4t} =  \left(1 - O(1/n)\right)$. 
     Now let us bound the $\Pr[\Event]$. 
        \[ \Pr\left[(w_1 \in S_L) \wedge (w_2 \in S_R) \right]  = \frac{t}{4t + 2} \cdot \frac{t}{4t + 1} = \Omega(1). \]
        Moreover, conditioning on the choice of $w_1$ and  $w_2$, by the memoryless property of exponential distribution, we have
        \begin{align*}
            \Pr[T_{w_2}>T_{w_3}+2] &= e^{-2 \eps} = \Omega(1),\\
            \Pr[T_{w_1} <  T_{w_2}+1] &= 1 - e^{- \eps} = \Omega(\eps). 
        \end{align*}
        Hence  $\Pr[\Event] = \Omega(\eps)$, as required.
    \end{proof}

    While the above families of graphs have low diameter, note that it is possible to increase the diameter to arbitrarily large by appending a long path. Moreover, only about $O(\log n)$ vertices on the path (i.e.,~vertices falling in the $O(\log n)$-radius neighborhood of the constructed graphs) could affect the outcome of the algorithm. By the analysis of \cref{claim:EN,claim:MPX}, these $O(\log n)$ vertices affect the outcome of the algorithm only if their $T$-values fall into the top three $T$-values among all  vertices, which happens with probability at most $O\left(\frac{\log n}{n}\right)$. 
}{}


\end{document}